\documentclass[12pt,reqno]{amsart}

\usepackage[margin=1.0in]{geometry}
\usepackage{graphicx}
\usepackage{amsmath}
\usepackage{amsfonts}
\usepackage{amssymb}
\usepackage{mathrsfs}
\usepackage[normalem]{ulem}
\usepackage{caption}
\usepackage{subcaption}
\usepackage{amsthm}
\usepackage{textcomp}
\usepackage{hyperref}
\usepackage{slashed}
\usepackage{mathtools}

\newtheorem*{proposition*}{Proposition}
\newtheorem*{theorem*}{Theorem}
\newtheorem*{conjecture*}{Conjecture}

\newtheorem*{claim*}{Claim}
\newtheorem*{lemma*}{Lemma}
\newtheorem*{corollary*}{Corollary}

\newtheorem{theorem}{Theorem}[section]
\newtheorem{thm}{Theorem}
\newtheorem{thmx}{Theorem}

\newtheorem{proposition}[theorem]{Proposition}

\newtheorem{lemma}[theorem]{Lemma}
\newtheorem{corollary}[theorem]{Corollary}

\newtheorem{conj}[thm]{Conjecture}
\newtheorem{conjx}[thmx]{Conjecture}

\newtheorem*{definition*}{Definition}
\newtheorem{definition}{Definition}[section]
\newtheorem*{assumption*}{Assumption}

\newtheorem*{remark*}{Remark}
\newtheorem{remark}{Remark}[section]

\newcommand{\R}{\mathbb{R}}
\newcommand{\s}{\mathbb{S}}

\newcommand{\N}{\mathbb{N}}

\newcommand{\snabla}{\slashed{\nabla}}
\newcommand{\sg}{\slashed{g}}

\numberwithin{equation}{section}
\allowdisplaybreaks

\begin{document}

\title{Linear waves in the interior of extremal black holes I}
\author{Dejan Gajic}
\address{University of Cambridge, Department of Applied Mathematics and Theoretical Physics, Wilberforce Road, Cambridge CB3 0WA, United Kingdom}
\email{D.Gajic@damtp.cam.ac.uk}
\date{}
\begin{abstract}
We consider solutions to the linear wave equation in the interior region of extremal Reissner--Nordstr\"om black holes. We show that, under suitable assumptions on the initial data, the solutions can be extended continuously beyond the Cauchy horizon and moreover, that their local energy is finite. This result is in contrast with previously established results for subextremal Reissner--Nordstr\"om black holes, where the local energy was shown to generically blow up at the Cauchy horizon.
\end{abstract}
\maketitle
\tableofcontents

\section{Introduction}
One of the most striking features of black hole spacetimes is the possible presence of spacetime singularities in their interior. The nature of singularities remains generally poorly understood, however, even when one restricts to vacuum spacetimes arising from small perturbations of Kerr initial data. 

In order to better understand black hole interiors, we consider the simplest toy model for the Einstein equations, governing the dynamics of perturbations of Kerr initial data, the linear wave equation on a fixed background spacetime:
\begin{equation}
\label{eq:waveqkerr}
\square_{g}\phi=0.
\end{equation}
Here, we either take $g=g_{a,M}$, the metric of a Kerr spacetime with mass $M$ and angular momentum parameter $a$, with $|a|\leq M$, or $g=g_{e,M}$, the metric of a Reissner--Nordstr\"om spacetime with mass $M$ and charge $e$, with $|e|\leq M$. We will view Reissner--Nordstr\"om as a ``poor man's version'' of Kerr.

In this paper, we initiate the mathematical study of (\ref{eq:waveqkerr}) in the interior region of \emph{extremal} black holes. Specifically, we consider (\ref{eq:waveqkerr}) in extremal Reissner--Nordstr\"om ($|e|=M$). Aretakis proved decay in time for solutions $\phi$ and their tangential derivatives along the event horizon in \cite{Aretakis2011a,Aretakis2011}, starting from Cauchy data on a spacelike hypersurface intersecting the event horizon. We first show that, under the assumption of the decay rates from \cite{Aretakis2011a,Aretakis2011}, $\phi$ is continuously extendible beyond the Cauchy horizon (the inner horizon). 

More surprisingly, by assuming \emph{stronger} decay for the tangential derivatives, we also show that the energy of $\phi$ remains finite, with respect to any uniformly timelike vector field and spacelike or null hypersurface intersecting the Cauchy horizon. The required decay along the event horizon is proved in upcoming work with Angelopoulos and Aretakis \cite{Angelopoulos2015}. The finiteness of the local energy of $\phi$ at the Cauchy horizon stands in sharp contrast to the interior region of \emph{subextremal} Reissner--Nordstr\"om, where the corresponding energy of $\phi$ has been shown to blow up for generic data \cite{Luk2015}. This distinctive property of the extremal Reissner--Nordstr\"om interior was first observed by Murata--Reall--Tanahashi in a numerical setting \cite{mureta1}.

By restricting to spherically symmetric solutions, we show that $\phi$ can in fact be extended as a $C^2$ function beyond the Cauchy horizon, if we assume precise asymptotics of $\phi$ that are suggested by the numerics in \cite{Lucietti2013}. 

We will show in a subsequent paper \cite{Gajic2015a} that under analogous assumptions for the decay in time of $\phi$ along the event horizon of extremal Kerr ($|a|=M$) and \emph{assuming axisymmetry} of $\phi$, we can extend $\phi$ continuously beyond the Cauchy horizon in the corresponding black hole interior. Moreover, the local energy of $\phi$ is finite at the Cauchy horizon. We will also show that the axisymmetry assumption on $\phi$ can be dropped in slowly rotating extremal Kerr--Newman spacetimes.

\subsection{Previous results for the linear wave equation on black hole backgrounds}
A \emph{qualitative} understanding of the behaviour of solutions to (\ref{eq:waveqkerr}) in the interior requires a \emph{quantitative} understanding of the solutions in the exterior. That is to say, continuous extendibility in the interior depends on precise decay rates for the solutions along the event horizon. We will review in this section decay results for the wave equation (\ref{eq:waveqkerr}) in the exterior of extremal and subextremal members of the Kerr--Newman family. We will also discuss boundedness and blow-up results in their interior regions.

Polynomial decay in time for solutions to (\ref{eq:waveqkerr}) in the exterior was obtained for the full subextremal range $|a|<M$ of Kerr spacetimes by Dafermos--Rodnianski--Shlapentokh-Rothman in \cite{Dafermos2014c}. This result is the culmination of many partial results by various different authors over the past decades. See \cite{dafrod5,Dafermos2014c} for comprehensive lists of references. The results of \cite{Dafermos2014c} have been generalised to subextremal Kerr--Newman in \cite{Civin2014}.

The geometry of the exterior region of \emph{extremal} Kerr, where $|a|=M$, exhibits important qualitative differences from the geometry of subextremal Kerr. Most notably, the local red-shift effect at the event horizon is absent in extremal Kerr. This is manifested in the study of (\ref{eq:waveqkerr}) as an \emph{Aretakis instability}, and was discovered by Aretakis in extremal Reissner--Nordstr\"om \cite{Aretakis2011a,Aretakis2011} and extremal Kerr \cite{are4}.

Aretakis showed more generally that non-decay of solutions or their transversal derivatives along the event horizon of extremal Kerr is a consequence of the existence of conserved quantities that do not vanish for generic data, the \emph{Aretakis constants}. He showed moreover that, under the assumption of pointwise decay of solutions and their tangential derivatives along the event horizon in affine time $v$, second-order transversal derivatives even \emph{blow up} as $v\to \infty$. Quantitative decay estimates for axisymmetric solutions in extremal Kerr were proved in \cite{are6,are5}. Conserved quantities along the event horizon can also arise for higher-spin equations on extremal Kerr \cite{lure1}. Boundedness and decay statements for (\ref{eq:waveqkerr}) on extremal Kerr \emph{without} the assumption of axisymmetry of the solutions remain open problems.

By using the previously established quantitative decay rates for solutions to (\ref{eq:waveqkerr}) along the event horizon of subextremal Reissner--Nordstr\"om spacetimes ($|e|<M$), it was shown that solutions are uniformly bounded in the interior and can be extended in $C^0$ beyond the Cauchy horizon.
\begin{thmx}[Franzen \cite{Franzen2014}]
\label{thm:franz}
Let $\phi$ be a solution to (\ref{eq:waveqkerr}) in \emph{subextremal} Reissner--Nordstr\"om arising from sufficiently regular and suitably decaying data on a spacelike hypersurface $\Sigma$ that is asymptotically flat at two ends. Then there exists a constant $C=C(M,\Sigma)>0$ and a constant $D_0>0$ that depends on the initial data, such that
\begin{equation*}
|\phi|\leq C D_0,
\end{equation*}
everywhere in the future domain of development of $\Sigma$. Moreover, $\phi$ admits a $C^0$ extension beyond the Cauchy horizon.
\end{thmx}
A weaker result, the statement that the above holds for \emph{fixed} spherical harmonic modes, can be deduced from the results of McNamara in \cite{Mcnamara1978a}, together with established quantitative decay results along the event horizon in the exterior of Kerr--Newman \cite{Civin2014}.

In the above subextremal setting there is also blow-up present in the interior. In particular, for generic initial data, $\phi$ cannot be extended as a function in $H^1_{\textnormal{loc}}$ beyond the Cauchy horizon.

\begin{thmx}[Luk-Oh \cite{Luk2015}]
\label{thm:lukoh}
Let $\phi$ be a solution to (\ref{eq:waveqkerr}) in \emph{subextremal} Reissner--Nordstr\"om arising from generic, smooth, compactly supported data on a spacelike hypersurface $\Sigma$ that is asymptotically flat at two ends. Then $\phi$ fails to be in $H^1_{\textnormal{loc}}$ in a neighbourhood of each point on the future Cauchy horizon.
\end{thmx}
\begin{figure}[b]
\begin{center}
\includegraphics[width=2.2in]{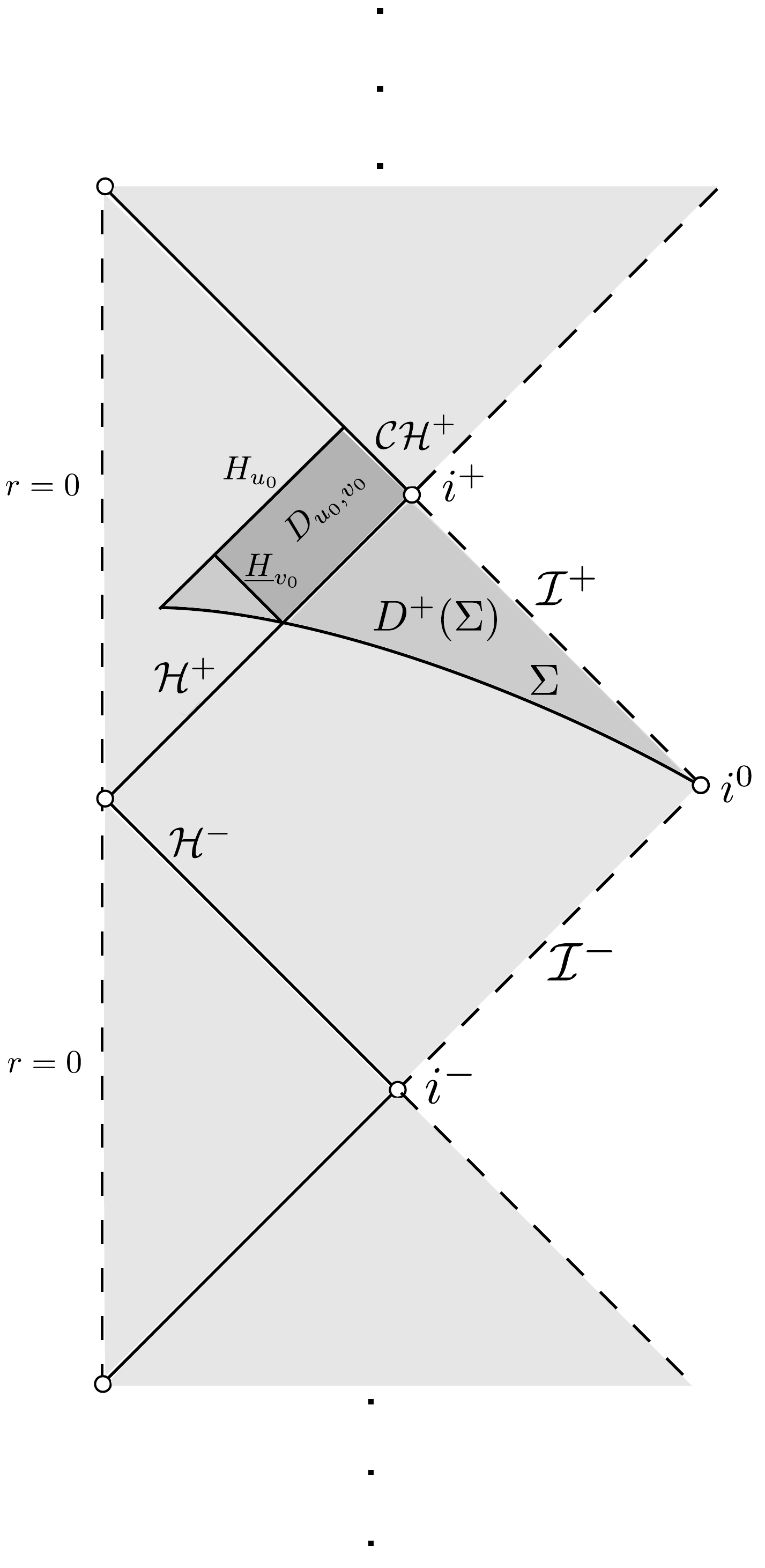}
\caption{\label{fig:fullspacetime} The Penrose diagram of maximal analytically extended extremal Reissner--Nordstr\"om.}
\end{center}
\end{figure}

Blow-up for derivatives of solutions to (\ref{eq:waveqkerr}) at the Cauchy horizon of subextremal Reissner--Nordstr\"om was first proved by McNamara \cite{Mcnamara1978} under an assumption on the event horizon. See also earlier numerical work by Simpson--Penrose \cite{Simpson1973}. Later, Poisson--Israel considered spherically symmetric nonlinear models describing ingoing and outgoing shells of radiation \cite{Poisson1990} and showed that the Kretschmann scalar blows up at the Cauchy horizon, due to the blow-up of a local mass quantity. The phenomenon of ``mass-inflation'' and the consequent blow-up of the Kretschmann scalar were proved by Dafermos in a more general setting \cite{Dafermos2003a,Dafermos2005c}, by considering the non-linear, spherically symmetric Einstein--Maxwell-scalar field model. See Section \ref{sec:nonlinresults} for more details.

\subsection{The linear wave equation in the interior of extremal Reissner--Nordstr\"om}
\label{sec:discussionresults}
In this paper, we impose Cauchy data for (\ref{eq:waveqkerr}) on a spacelike hypersurface $\Sigma$ in extremal Reissner--Nordstr\"om. It is appropriate to consider a hypersurface $\Sigma$ that is asymptotically flat at one end and intersects the black hole interior. Due to the geometry of the interior, $\Sigma$ must necessarily be incomplete. We restrict to the future domain of dependence of $\Sigma$, denoted by $D^+(\Sigma)$. The event horizon is denoted by $\mathcal{H}^+$ and we denote the segment of the future Cauchy horizon emanating from future timelike infinity $i^+$ in the interior by $\mathcal{CH}^+$. See Figure \ref{fig:fullspacetime}. For convenience, we will also denote the entire inner horizon of extremal Reissner--Nordstr\"om by $\mathcal{CH}^+$.  

We will often restrict to a subset $D_{u_0,v_0}$ of $D^+(\Sigma)\cup \mathcal{CH}^+$ that is the intersection of the causal future of the outgoing null hypersurface $\mathcal{H}^+$, the causal future of the ingoing null hypersurface $\uline{H}_{v_0}$ in the interior and the causal past of the outgoing null hypersurface $H_{u_0}$ in the interior. Here, $u$ and $v$ are Eddington--Finkelstein double null coordinates, $\uline{H}_{v_0}$ is a constant $v$ hypersurface in $\mathcal{M}\cup \mathcal{H}^+$ that intersects $\mathcal{H}^+$ and $H_{u_0}$ is a constant $u$ hypersurface in $\mathcal{M}$. We take $D_{u_0,v_0}$ to include a segment of $\mathcal{CH}^+$. See Figure \ref{fig:fullspacetime}. We can express $D_{u_0,v_0}$ as the set
\begin{equation*}
D_{u_0,v_0}=\{x\in D^+(\Sigma)\cup\mathcal{CH}^+\::\: U(x)\in[0,U(u_0)],\:V(x)\in[V(v_0),M],\:(U(x),V(x))\neq (0,M)\},
\end{equation*}
with $U$ an ingoing double-null coordinate that is regular at $\mathcal{H}^+$, where $U=0$, and $V$ an outgoing double-null coordinate that is regular at $\mathcal{CH}^+$, where $V=M$.

We can view $\phi$ restricted to $\mathcal{H}^+\cap\{v\geq v_0\}$ and $\uline{H}_{v_0}\cap \{u\leq u_0\}$, arising from Cauchy data on $\Sigma$, as characteristic initial data, in order to decouple the analysis in the interior from the analysis in the exterior.

We will first of all show that an analogous result to Theorem \ref{thm:franz} holds in extremal Reissner--Nordstr\"om.
\begin{thm}($L^{\infty}$-boundedness and $C^0$-extendibility)
\label{thm:linftyboundv1}
Let $\phi$ be a solution to (\ref{eq:waveqkerr}) in \emph{extremal} Reissner--Nordstr\"om arising from suitably decaying data on $\Sigma$. Then there exists a constant $C=C(M,\Sigma)>0$ and a constant $D_0>0$ that depends on the initial data, such that
\begin{equation*}
|\phi|\leq CD_0,
\end{equation*}
everywhere in $D^+(\Sigma)$. Moreover, $\phi$ admits a $C^0$ extension beyond $\mathcal{CH}^+$.
\end{thm}

In view of the decay results along $\mathcal{H}^+$ in \cite{Aretakis2011a,Aretakis2011}, it is sufficient to prove the results of Theorem \ref{thm:linftyboundv1}, starting from appropriate characteristic initial data.

We will also show that, \textbf{in contrast with Theorem \ref{thm:lukoh}}, $\phi$ can be extended beyond $\mathcal{CH}^+$ as a continuous function in $H^1_{\textnormal{loc}}$ in the extremal case. We first formulate a theorem where characteristic initial data are posed on $\mathcal{H}^+$.
\begin{thm}[$H^1_{\textnormal{loc}}$-extendibility; first version]
\label{thm:H1boundv1}
Consider suitably regular characteristic initial data on $\uline{H}_{v_0}\cup \mathcal{H}^+$, such that
\begin{equation}
\label{eq:initialdecH+}
\int_{\mathcal{H}^+\cap\{v\geq v_0\}} v^2(\partial_v\phi)^2+|\snabla\phi|^2<\infty.
\end{equation}
Then $\phi$ can be extended beyond $\mathcal{CH}^+$ as a $H^1_{\textnormal{loc}}$ function.
\end{thm}

The requirement (\ref{eq:initialdecH+}) follows from \emph{improved} decay results in the exterior of extremal Reissner--Nordstr\"om for solutions arising from suitable Cauchy data on $\Sigma$, compared to the decay results in \cite{Aretakis2011a,Aretakis2011}. See upcoming work with Angelopoulos and Aretakis \cite{Angelopoulos2015}. We can therefore reformulate Theorem \ref{thm:H1boundv1} in terms of initial Cauchy data on $\Sigma$.

\begin{thm}[$H^1_{\textnormal{loc}}$-extendibility; second version, using \cite{Angelopoulos2015}]
\label{thm:H1boundv2}
Let $\phi$ be a solution to (\ref{eq:waveqkerr}) in extremal Reissner--Nordstr\"om arising from suitably regular and decaying data on $\Sigma$. Then $\phi$ can be extended beyond $\mathcal{CH}^+$ as a $H^1_{\textnormal{loc}}$ function.
\end{thm}

It was first suggested by Murata--Reall--Tanahashi in \cite{mureta1} that spherically symmetric $\phi$ should be extendible in $H_{\textnormal{loc}}^1$ (and in fact, in $C^1$) beyond the Cauchy horizon, in the context of perturbations of extremal Reissner--Nordstr\"om in the  nonlinearly coupled spherically symmetric Einstein--Maxwell-scalar field model. See the discussion in Section \ref{sec:nonlinresults}.

\subsubsection{Late-time tails along the event horizon of extremal Reissner--Nordstr\"om}
\label{sec:priceslaw}
We can further refine the results of Theorem \ref{thm:H1boundv1} and \ref{thm:H1boundv2}, obtaining \emph{pointwise} estimates for the derivatives of $\phi$ at the Cauchy horizon, if we have a more precise picture of the decay of $\phi$ in $v$ along $\mathcal{H}^+$, where $v$ is an outgoing Eddington--Finkelstein double-null coordinate. In this section, we discuss the asymptotics of $\phi$ in $v$ that are expected to hold along the event horizon of subextremal and extremal Reissner--Nordstr\"om.

Recall that in Minkowski space, one can show arbitrarily fast decay of $\phi$ in $t$ in a bounded region $\{r\leq R\}$ if it decays suitably fast initially in $r$. In particular, if $\phi$ is initially compactly supported, it will vanish in $\{r\leq R\}$ at late times.

In black hole spacetimes, one expects non-trivial asymptotic behaviour $\phi$ in time, even for compactly supported initial data. In subextremal Reissner--Nordstr\"om the expected behaviour of these ``late-time tails'' is governed by \emph{Price's law} \cite{Price1972}. In particular, Price's law predicts for fixed spherical harmonic modes $\phi_l$, with $l\in \N_0$, arising from compactly supported initial data on a spacelike Cauchy hypersurface, that
\begin{equation}
\label{eq:priceslawsub}
|\phi_l|\sim v^{-3-2l}\quad \textnormal{and}\quad |\partial_v\phi_l|\sim v^{-4-2l},
\end{equation}
along the event horizon, where the constants in (\ref{eq:priceslawsub}) are generically non-vanishing.

Price's law (\ref{eq:priceslawsub}) gives in particular an \emph{upper} bound for the decay of $\phi$ along the event horizon of subextremal Reissner--Nordstr\"om. Rigorous results pertaining to the upper bound in (\ref{eq:priceslawsub}) have been obtained in \cite{Dafermos2005a,Donninger2011a,Tataru2009,dosch1,Metcalfe2012}. Moreover, Luk--Oh showed in \cite{Luk2015} that $\phi$ cannot decay with a polynomial rate faster than $v^{-3}$ along the event horizon of subextremal Reissner--Nordstr\"om, for generic, compactly supported data. This fact is made use of in Theorem \ref{thm:lukoh}.

The methods in \cite{Dafermos2005a,Donninger2011a,Tataru2009,dosch1,Metcalfe2012,Luk2015} break down in extremal Reissner--Nordstr\"om, in view of the absence of the local red-shift effect. Heuristics and numerics regarding late-time tails for extremal Reissner--Nordstr\"om in \cite{Lucietti2013,ori2} suggest an extremal variant of ``Price's law'' that in particular predicts:
\begin{equation}
\label{eq:priceslaw}
|\phi_l|\sim v^{-(l+1)}\quad \textnormal{and}\quad |\partial_v\phi_l|\sim v^{-(l+2)},
\end{equation}
along $\mathcal{H}^+$, where the constants in (\ref{eq:priceslaw}) are non-vanishing for generic, compactly supported initial data on $\Sigma$.

The decay rates appearing in Price's law (\ref{eq:priceslaw}) are related to the existence of conserved quantities along $\mathcal{I}^+$ and $\mathcal{H}^+$. The vanishing of these conserved quantities affects the asymptotics of $\phi$. We first define the conserved quantities.

It turns out that the following limit in outgoing Eddington--Finkelstein coordinates $(u,r,\theta,\varphi)$, if well-defined, is independent of $u$:
\begin{equation*}
\lim_{r\to \infty}\int_{\s^2}r^2\partial_r(r\phi)(u,r,\theta,\varphi)\,d\mu_{\s^2}= 4\pi I_0[\phi],
\end{equation*}
where $I_0[\phi]$ is a constant, determined by the initial data on $\Sigma$. The constant $I_0[\phi]$ is known as the \emph{first Newman--Penrose constant} \cite{Newman1968}. See also \cite{Lucietti2013,Aretakis2013}. 

Similarly, the following limit in ingoing Eddington--Finkelstein coordinates $(v,r,\theta,\varphi)$ is independent of $v$:
\begin{equation*}
\lim_{r\to M}\int_{\s^2}\partial_r(r\phi)(v,r,\theta,\varphi)\,d\mu_{\s^2}= 4\pi M H_0[\phi],
\end{equation*}
where $H_0[\phi]$ is a constant, determined by the initial data on $\Sigma$. The constant $H_0[\phi]$ is known as the \emph{zeroth Aretakis constant} \cite{Aretakis2011}.

For generic, compactly supported initial data, $I_0[\phi]=0$, but $H_0[\phi]\neq 0$. If additionally, the data are supported \emph{away} from $\mathcal{H}^+$, then $H_0[\phi]=0$. Solutions arising from initial data such that \emph{both} $H_0[\phi]=0$ and $I_0[\phi]=0$ are considered in a numerical setting in \cite{Lucietti2013}. In this case, the following late-time tails are suggested:
\begin{equation}
\label{eq:priceslawvan}
|\phi_l|\sim v^{-(l+2)}\quad \textnormal{and}\quad |\partial_v\phi_l|\sim v^{-(l+3)},
\end{equation}
where the constants in (\ref{eq:priceslawvan}) are non-vanishing for generic initial data on $\Sigma$ with $H_0[\phi]=0$ and $I_0[\phi]=0$. 

See also upcoming work with Angelopoulos and Aretakis on late-time tails in spherically symmetric black hole backgrounds \cite{Angelopoulos2015temp}.

One can consider the asymptotic behaviour of $\phi$ along $\mathcal{H}^+$ \emph{beyond} Price's law as stated in (\ref{eq:priceslaw}) by including the next-to-leading order term in $v^{-1}$. The numerical analysis in \cite{Lucietti2013} suggests that the leading-order terms for the spherically symmetric mode $\phi_0$ along $\mathcal{H}^+$, arising from compactly supported initial data on $\Sigma$, are given by
\begin{equation*}
- MH_0[\phi]v^{-1}+M^3H_0[\phi]v^{-2}\log v.
\end{equation*}
Motivated by the above terms, we will introduce the following slightly stronger assumptions for the asymptotic behaviour of $\phi_0$ along $\mathcal{H}^+$:
\begin{align}
\label{eq:beyondpriceslaw}
&\phi_0|_{\mathcal{H}^+}(v)=- MH_0[\phi]v^{-1}+M^3H_0[\phi]v^{-2}\log v+\mathcal{O}_2(v^{-2}),\\
\label{eq:addassumpinitial}
&\lim_{v\to \infty}\left[v^2\partial_v(v^2\partial_v\phi_0)|_{\mathcal{H}^+}(v)-2M^3H_0[\phi]\log v\right]\: \textnormal{is well-defined}.
\end{align}
Here, we use the notation $\mathcal{O}_2$ to group the higher-order terms in $v^{-1}$, i.e.\ all terms in $\mathcal{O}_2(v^{-2})$ decay at least as fast as $v^{-2}$ and $k$-th order derivatives, up to $k=2$, decay at least as fast as $v^{-2-k}$.\footnote{See also \cite{Gomez1994}, where similar asymptotics are predicted for the radiation fields of spherically symmetric solutions to (\ref{eq:waveqkerr}) along $\mathcal{I}^+$ in Schwarzschild, with respect to an ingoing Eddington--Finkelstein null coordinate $u$, arising from initial data with non-vanishing first Newman--Penrose constant imposed along an outgoing null hypersurface with additional reflective boundary conditions on the surface $\{r=R\}$, with $R>2M$.}

Note that the asymptotics (\ref{eq:beyondpriceslaw}) ensure that the quantity $v^2\partial_v(v^2\partial_v\phi_0)|_{\mathcal{H}^+}(v)-2M^3H_0[\phi]\log v$ in (\ref{eq:addassumpinitial}) is uniformly bounded in $v$, but do not guarantee that the limit as $v\to\infty$ is well-defined.

\subsubsection{Pointwise estimates for first-order and second-order derivatives}
\label{sec:intropointwiseest}
We can obtain pointwise estimates for $\phi$ and its derivatives under stronger assumptions on the asymptotics of $\phi$ along $\mathcal{H}^+$ than those required for Theorem \ref{thm:H1boundv1}. 

First of all, we show that spherically symmetric $\phi=\phi_0$ are extendible in $C^1$ beyond $\mathcal{CH}^+$, under an assumption that is compatible with the upper bound in (\ref{eq:priceslaw}). We consider, as in Theorem \ref{thm:H1boundv1}, characteristic initial data along $\mathcal{H}^+$.
\begin{thm}[$C^1$-extendibility of $\phi_0$]
\label{thm:c1boundphi0}
Consider suitably regular \uline{spherically symmetric} characteristic initial data on $\mathcal{H}^+\cup \uline{H}_{v_0}$, such that $\lim_{v\to \infty} v^2 \partial_v\phi_0|_{\mathcal{H}^+}(v)$ is well-defined. Then the arising solution $\phi_0$ can be extended as a $C^1$ function beyond $\mathcal{CH}^+$.
\end{thm}
Note that the assumption that $\lim_{v\to \infty} v^2 \partial_v\phi_0|_{\mathcal{H}^+}(v)$ is well-defined implies the uniform decay estimate $|\partial_v\phi_0|_{\mathcal{H}^+}|\lesssim v^{-2}$, which is a \emph{stronger} statement than the assumption (\ref{eq:initialdecH+}) of Theorem \ref{thm:H1boundv1} for $\phi_0$.

For solutions $\phi$ of (\ref{eq:waveqkerr}) \emph{without} any symmetry restrictions, we can still show extendibility beyond $\mathcal{CH}^+$ in the H\"older space $C^{0,\alpha}$, with $\alpha<1$, if we assume decay rates along $\mathcal{H}^+$ that are consistent with the upper bound in Price's law (\ref{eq:priceslaw}).

 \begin{thm}[$C^{0,\alpha}$-extendibility of $\phi$]
\label{thm:c11boundphi}
Let $0\leq \alpha<1$. Consider suitably regular characteristic initial data on $\mathcal{H}^+\cup \uline{H}_{v_0}$, such that along $\mathcal{H}^+$
\begin{align*}
\sum_{|k|\leq 2}\left(\int_{\s^2}|\partial_v\snabla^k\phi|^2\,d\mu_{\s^2}\right)(v)&\leq C_0v^{-2(\alpha+1)},\\
\sum_{k=1}^4 \left( \int_{\s^2} |\snabla^k\phi|^2\,d\mu_{\s^2}\right)(v)&\leq Cv^{-2\alpha}.
\end{align*}
for some constant $C_0>0$, where $\snabla^k$ denotes $k$-th order angular derivatives. 

Then the arising solution $\phi$ can be extended as a $C^{0,\alpha}$ function beyond $\mathcal{CH}^+$.
\end{thm}

Furthermore, under the more precise assumption for the asymptotics for $\phi_0$ along the horizon, given in (\ref{eq:beyondpriceslaw})  and (\ref{eq:addassumpinitial}), which are predicted to hold for suitably regular and compactly supported data on $\Sigma$, we show that $\phi_0$ is extendible as a $C^2$ function beyond $\mathcal{CH}^+$.

\begin{thm}[$C^2$-extendibility of $\phi_0$]
\label{thm:c2blowup}
Consider suitably regular \uline{spherically symmetric} characteristic initial data on $\mathcal{H}^+\cup \uline{H}_{v_0}$, such that the asymptotics in (\ref{eq:beyondpriceslaw}) and (\ref{eq:addassumpinitial}) hold. Then the arising solution $\phi_0$ can be extended as a $C^2$ function beyond $\mathcal{CH}^+$.
\end{thm}

In the case that $H_0\neq 0$, given the leading-order behaviour of $\phi_0$ along $\mathcal{H}^+$ in (\ref{eq:beyondpriceslaw}), the proof of the above theorem requires the next-to-leading order term in the asymptotics of $\phi_0$ to correspond \emph{exactly} with the next-to-leading order term in (\ref{eq:beyondpriceslaw}). We moreover show that any deviation in the next-to-leading order term in the asympttotics leads instead to \uline{blow-up} of the second-order transversal derivatives of $\phi_0$ at $\mathcal{CH}^+$. If $H_0=0$, the precise terms in the asymptotics (\ref{eq:beyondpriceslaw}) are not important to conclude that $\phi_0$ is extendible in $C^2$ beyond $\mathcal{CH}^+$. 

In comparison, the results of Dafermos in \cite{Dafermos2005c} imply in particular the blow-up of the first-order derivatives of $\phi_0$ at each point on the Cauchy horizon in subextremal Reissner--Nordstr\"om, if 
Price's law (\ref{eq:priceslawsub}) for \uline{sub}extremal Reissner--Nordstr\"om is assumed as a \emph{lower} bound for $\phi_0$ (in fact, a weaker lower bound is sufficient).

\subsection{Main ideas in the proofs of Theorem \ref{thm:linftyboundv1} and \ref{thm:H1boundv1}}
\label{sec:mainideasThm12}
We will illustrate in this section the main ideas in the proofs of Theorem \ref{thm:linftyboundv1} and \ref{thm:H1boundv1}.

In order to obtain the required $L^2$ estimates in the interior, we consider appropriately weighted energies along null hypersurfaces. It is natural to construct the corresponding energy currents with respect to suitable vector fields (vector field multipliers) and to moreover act with vector fields (commutation vector fields) on solutions to (\ref{eq:waveqkerr}). This kind of construction is known in the literature as ``the vector field method''; see \cite{Klainerman2010}. We discuss the vector field method in more detail in Section \ref{sec:vectorfieldmethod}, in the setting of the interior of extremal Reissner--Nordstr\"om. Weighted energy estimates derived in Section \ref{sec:eestimatesrn} allow us to obtain $L^2$ estimates on the 2-spheres foliating the null hypersurfaces, which in turn lead to $L^{\infty}$ estimates, by applying standard Sobolev estimates on the 2-spheres; see Section \ref{sec:pointwise}.

The absence of the local red-shift effect (i.e.\ the vanishing of the surface gravity) in the extremal case affects the choice of vector fields that can be used to construct weighted energies in the interior. In particular, as shown in \cite{Aretakis2011a} in a neighbourhood of the event horizon in the exterior of extremal Reissner--Nordstr\"om, there exists no timelike vector field such that the spacetime term in the corresponding energy estimates has the right sign at the event horizon (a ``red-shift vector field''). Thus, decay along the event horizon cannot be propagated to spacelike hypersurfaces in a neighbourhood of the horizon in the interior by employing a red-shift vector field, as was done in Schwarzschild \cite{Luk2010} and later applied to subextremal Reissner--Nordstr\"om in \cite{Franzen2014}.

There is another important geometric difference between the interior of extremal and subextremal Reissner--Nordstr\"om: the 2-spheres in the extremal interior that are preserved under isometries of spherical symmetry are not trapped. In terms of the area radius $r$ of the spheres foliating the spacetime and Eddington--Finkelstein double-null coordinates $(u,v)$, this means that
\begin{equation*}
\partial_u r<0,\quad \partial_vr\geq 0,
\end{equation*}
in the interior, whereas in subextremal Reissner--Nordstr\"om the above signs coincide. The difference in the sign of $\partial_vr$ in particular means that the spacetime terms in the energy estimates corresponding to vector fields of the form
\begin{equation*}
r^q(\partial_u+\partial_v),
\end{equation*}
containing all derivatives of $\phi$, do not have a favourable sign, in contrast with subextremal Reissner--Nordstr\"om (for $q>0$ sufficiently large). Both the above vector field and the red-shift vector field are crucial for the arguments in \cite{Franzen2014}.

In our extremal case we can, regardless of the above differences, do energy estimates \emph{directly} with respect to a uniformly timelike vector field of the form
\begin{equation*}
u^2\partial_u+v^2\partial_v.
\end{equation*}
See (\ref{def:mainvectorfieldN}) for a definition of the full family of vector fields that we will consider. We are able to control all the spacetime error terms in the corresponding energy estimates by absorbing them into the energy fluxes along null hypersurfaces. Similar vector fields are also used in the subextremal case, but only sufficiently close to the Cauchy horizon. The reason why they can be used everywhere in a neighbourhood of timelike infinity $i^+$ in the interior of extremal Reissner--Nordstr\"om is related to a third important difference between the extremal and subextremal cases: the metric component $g_{uv}$ in Eddington--Finkelstein double-null coordinates decays uniformly in every direction towards $i^+$ in the extremal case, as shown in Section \ref{sec:geom}, whereas $g_{uv}$ is constant along constant $r$ hypersurfaces that approach $i^+$ in the subextremal case.

\subsection{Main ideas in the proofs Theorem \ref{thm:c1boundphi0}, \ref{thm:c11boundphi} and \ref{thm:c2blowup}}
In this section we will describe the main ideas involved in obtaining the pointwise results in the theorems of Section \ref{sec:intropointwiseest}.

By commuting (\ref{eq:waveqkerr}) with ingoing null vector fields and angular momentum operators, we can extend the results of Theorem \ref{thm:linftyboundv1} to obtain uniform pointwise boundedness and $C^0$-extendibility of arbitrarily many derivatives of $\phi$ that are tangential to $\mathcal{CH}^+$; see Section \ref{sec:pointwise}. If we restrict to spherically symmetric $\phi$ we can also obtain uniform pointwise boundedness of a regular outgoing transversal derivative $\partial_V\phi$ at $\mathcal{CH}^+$,\footnote{Note that \emph{a priori} uniform boundedness of $\partial_V\phi$ only implies boundedness \emph{at each point} along $\mathcal{CH}^+$ of $\partial_r\phi$, with respect to outgoing Eddington-Finkelstein coordinates $(u,r, \theta, \varphi)$ that can be extended beyond $\mathcal{CH}^+$. However, as the quantity $\partial_r(r\phi)$ is conserved along $\mathcal{H}^+$ (cf. the conserved Aretakis constants along $\mathcal{H}^+$) and $\phi$ is uniformly bounded, it follows that $\partial_r\phi|_{\mathcal{CH}^+}$ must in fact be uniformly bounded.} by writing the wave equation (\ref{eq:waveqkerr}) as a transport equation and integrating in the ingoing null direction, using boundedness of a weighted $L^1$ norm for $\partial_u\phi$ along ingoing radial null geodesics. This is done in Section \ref{sec:weightedL1} and \ref{sec:Linftyfirstorder}. We can show similarly that $\partial_V\phi$ can be extended as a continuous function beyond $\mathcal{CH}^+$ to conclude that $\phi$ can be extended as a $C^1$ function, if $\phi$ is spherically symmetric; see Section \ref{sec:regularityatCH+}.

For a solution $\phi$ that is not spherically symmetric, the above method fails; nevertheless we are still able to show in Section \ref{sec:decaydernosymm} that $\partial_V \phi$ blows up at most logarithmically in $v$ as we approach $\mathcal{CH}^+$. As a consequence, $\phi$ can be extended as a $C^{0,\alpha}$ function beyond $\mathcal{CH}^+$, for $\alpha<1$.

In the case of spherically symmetric $\phi$, we can also use the wave equation to propagate the asymptotic behaviour of $\partial_u \phi$ along $\uline{H}_{v_0}$ to all $\uline{H}_{v}$. For characteristic initial data with a vanishing Aretakis constant $H_0$, this asymptotic behaviour, together with the asymptotics (\ref{eq:beyondpriceslaw}) for $H_0=0$ along $\mathcal{H}^+$, is sufficient to infer uniform boundedness of $\partial_V^2\phi$ for spherically symmetric $\phi$; see Section \ref{sec:Linftysecondorder}.\footnote{As for the first-order derivative $\partial_V\phi$, uniform boundedness of $\partial_V^2\phi$ only guarantees boundedness of $\partial_r^2\phi$ at each point of $\mathcal{CH}^+$. In fact, assuming the asymptotics (\ref{eq:beyondpriceslaw}) with $H_0\neq 0$, it can be shown that $\partial_r^2\phi$ \emph{blows up} as we approach $i^+$ along $\mathcal{CH}^+$ (cf. the blow-up of $\partial_r^2\phi$ towards $i^+$ along $\mathcal{H}^+$ for $H_0\neq 0$ in \cite{Aretakis2011}.)} Here, we commute the wave equation (\ref{eq:waveqkerr}) with $\partial_V$ and rewrite it as a transport equation for $\partial_V^2\phi$. Recall moreover that arbitrary many tangential derivatives to $\mathcal{CH}^+$ are also uniformly bounded.

If we consider instead data with a non-vanishing $H_0$, we can still propagate the asymptotic behaviour of $\partial_u \phi$ along $\uline{H}_{v_0}$ to all $\uline{H}_{v}$. In this case, the constants appearing in front of the leading-order term and next-to-leading-order term in (\ref{eq:beyondpriceslaw}) become vital for the argument. That is to say, the asymptotic behaviour of $\partial_u \phi$ implies that the \emph{difference} $\partial_V^2\phi(u,v)-\partial_V^2\phi|_{\mathcal{H}^+}(v)$ must \emph{blow up} logarithmically in $v$, as we approach $\mathcal{CH}^+$. The only way of preventing $\partial_V^2\phi(u,v)$ from blowing up at $\mathcal{CH}^+$ is to require the leading-order term in the asymptotics of $\partial_V^2\phi|_{\mathcal{H}^+}(v)$ to cancel out \emph{precisely} the term in the difference $\partial_V^2\phi(u,v)-\partial_V^2\phi|_{\mathcal{H}^+}(v)$ that blows up logarithmically in $v$. Remarkably, it turns out that this cancellation indeed occurs if we assume the asymptotics (\ref{eq:beyondpriceslaw}). We can therefore conclude that $\partial_V^2\phi$ is bounded at $\mathcal{CH}^+$ for spherically symmetric $\phi$, even if $H_0\neq 0$. This is also done in Section \ref{sec:Linftysecondorder}.

Via similar methods, using again the crucial cancellation described in the above paragraph, we can show that $\partial_V^2\phi$ is continuous at $\mathcal{CH}^+$, for spherically symmetric $\phi$, to conclude that $\phi$ can be extended as a $C^2$ function beyond $\mathcal{CH}^+$; see Section \ref{sec:regularityatCH+}.

\subsection{The linear wave equation in the interior of Kerr(-Newman)}
In a subsequent paper \cite{Gajic2015a} we consider (\ref{eq:waveqkerr}) in the interior region of extremal Kerr. We show that analogous results to Theorem \ref{thm:linftyboundv1} and Theorem \ref{thm:H1boundv1} hold, if we restrict to axisymmetric solutions $\phi$ and assume suitable decay for $\phi$ along the event horizon. We can remove the axisymmetry restriction if we consider slowly rotating extremal Kerr--Newman instead of Kerr, i.e.\ we need $|a|$ to be suitably small compared to $M$. The extendibility of \emph{non}-axisymmetric $\phi$ beyond the Cauchy horizon in $C^0$ or $H^1_{\textnormal{loc}}$ remains an open problem in extremal Kerr. This illustrates how the analysis in the extremal Reissner--Nordstr\"om interior does not quite capture all of the difficulties present in the extremal Kerr interior.

Let us note that an analogue of Theorem \ref{thm:franz} can also be obtained in the subextremal Kerr interior. See upcoming work of Franzen \cite{franz2}. 

\subsection{Nonlinear results and conjectures in interior regions}
\label{sec:nonlinresults}
The linear wave equation (\ref{eq:waveqkerr}) on a fixed Reissner--Nordstr\"om or Kerr spacetime is the simplest toy model for the quantitative behaviour in the interior region of spacetimes arising from small perturbations of Kerr initial data, in the context of the Cauchy problem for the vacuum Einstein equations:
\begin{equation}
\label{eq:ee}
\textnormal{Ric}(g)=0.
\end{equation}

One strategy for exploring the effect of nonlinearities (``backreaction'') in (\ref{eq:ee}) is to impose the restriction of spherical symmetry and consider the nonlinearly coupled Einstein--Maxwell-scalar field system of equations. The Einstein equations simplify greatly in spherical symmetry, while the coupling with the wave equation and Maxwell's equations still allows for a large variety in global structures of the spacetimes.

The linear results for (\ref{eq:waveqkerr}) on a fixed Reissner--Nordstr\"om or Kerr background and the nonlinear results for the spherically symmetric Einstein--Maxwell-scalar field model can moreover be used to formulate conjectures regarding the global structure of the interior region of spacetimes arising from the evolution of perturbations of Kerr initial data for (\ref{eq:ee}) and the behaviour of metric components at the future spacetime boundaries.

\subsubsection{Results for the spherically symmetric Einstein--Maxwell-scalar field model}
We will discuss in this section the spherically symmetric Einstein--Maxwell-scalar field system of equations, which was studied by Dafermos in \cite{Dafermos2003a,Dafermos2005c,Dafermos2014a}. 

Dafermos showed that black hole solutions approaching a subextremal Reissner--Nordstr\"om solution along the event horizon\footnote{That is to say, the area radius of the spheres foliating the event horizon approaches a constant $r_+>|e|$.} have a non-empty Cauchy horizon, beyond which the metric can be extended as a $C^0$ function. Moreover, under a lower bound assumption on the decay of the scalar field along the event horizon, transversal derivatives of $\phi$ and the $L^2_{\textnormal{loc}}$ norm of the Christoffel symbols of the metric blow up at the Cauchy horizon. Dafermos' proof makes use of mass-inflation as a mechanism for blow-up.

Finally, if one restricts to spacetimes arising from small perturbations of two-ended subextremal Reissner--Nordstr\"om data, the entire future boundary of the interior region is in fact a bifurcate Cauchy horizon across which the metric is $C^0$-extendible but the $L^2_{\textnormal{loc}}$ norm of the Christoffel symbols blows up.

The spherically symmetric Einstein--Maxwell-scalar field system has recently also been studied in a similar context with a positive cosmological constant \cite{Costa2015,Costa2014,Costa2014a}. 

Furthermore, the spherically symmetric Einstein--Maxwell-scalar field system has been considered by Murata--Reall--Tanahashi in \cite{mureta1} from a numerical point of view, for spacetimes arising from ``outgoing'' characteristic initial data, where the outgoing initial null hypersurface is isometric to a null hypersurface in the exterior region of Reissner--Nordstr\"om. They found that for fine-tuned initial data the corresponding future developments are black hole spacetimes containing  no (marginally) trapped surfaces of symmetry, that approach extremal Reissner--Nordstr\"om along the event horizon.\footnote{Here, we mean that the area radius of the spheres foliating the event horizon approaches the constant $r_+=|e|$.}

Moreover, the numerics in \cite{mureta1} suggest that the interior region of these spacetimes has a non-empty Cauchy horizon across which the metric is extendible in $C^0$ \emph{with} Christoffel symbols in $L^2_{\textnormal{loc}}$ at the Cauchy horizon. Additionally, both the scalar field $\phi$ and \emph{all} its first-order derivatives remain bounded at the Cauchy horizon. The results of Theorem \ref{thm:c1boundphi0} suggest that this behaviour of the scalar field indeed holds. Moreover, the results of Theorem \ref{thm:c2blowup} suggest that in fact all second-order derivatives should remain bounded at each point along the Cauchy horizon.

\subsubsection{Conjectures for the vacuum Einstein equations}
Analogues of the linear results for solutions to (\ref{eq:waveqkerr}) in subextremal Reissner--Nordstr\"om, and the nonlinear spherically symmetric results in the interior region of ``asymptotically subextremal'' black hole spacetimes, are conjectured to hold when one considers small perturbations of subextremal Kerr initial data in the context of the Cauchy problem for the vacuum Einstein equations (\ref{eq:ee}).

Recently, Dafermos--Luk proved the following theorem \textbf{without symmetry assumptions} in a characteristic initial value problem formulation of (\ref{eq:ee}). 
\begin{thmx}[Dafermos--Luk \cite{Dafermos2015d}]
\label{thm:c0stabilitysubextr}
Consider characteristic initial data for (\ref{eq:ee}) on a bifurcate null hypersurface $\mathcal{H}^+\cup \widetilde{\mathcal{H}}^+$, where $\mathcal{H}^+$ and $\widetilde{\mathcal{H}}^+$ have future-affine complete null generators and their induced geometry is globally close to and dynamically approaches that of the event horizon of a Kerr spacetime at a sufficiently fast polynomial rate, with mass parameters $M$ and $\tilde{M}$, and rotation parameters $0 < |a| < M$ and $0 < |\tilde{a}| < \tilde{M}$, respectively. Then the maximal development can be extended beyond a bifurcate Cauchy horizon $\mathcal{CH}^+$ as a Lorentzian manifold with $C^0$ metric.
\end{thmx}

One can consider Theorem \ref{thm:franz} as a linear, toy model version of Theorem \ref{thm:c0stabilitysubextr}. Moreover, the nonlinear, spherically symmetric toy model version is contained in the results of \cite{Dafermos2003a,Dafermos2005c,Dafermos2014a}. 

Theorem \ref{thm:c0stabilitysubextr} is accompanied by a conjecture.
\begin{conjx}[\cite{Dafermos2015d}]
\label{conj:instabilitych}
\hspace{1pt}
\begin{itemize}\setlength\itemsep{1em}
\item[(i)]
The assumptions on $\mathcal{H}^+\cup \widetilde{\mathcal{H}}^+$ from Theorem \ref{thm:c0stabilitysubextr} hold for spacetimes arising from suitably small perturbations of two-ended asymptotically flat subextremal Kerr initial Cauchy data for (\ref{eq:ee}).
\item[(ii)] 
Under suitable additional assumptions on the induced geometry of $\mathcal{H}^+\cup \widetilde{\mathcal{H}}^+$ from Theorem
\ref{thm:c0stabilitysubextr}, $\mathcal{CH}^+$ is a weak null singularity, across which the metric is inextendible as a Lorentzian manifold with locally square integrable Christoffel symbols.

\item[(iii)] The additional assumptions in \textnormal{(ii)} hold for \uline{generic} asymptotically flat initial data in \emph{(i)}.
\end{itemize}
\end{conjx}

Theorem \ref{thm:lukoh} can be viewed as the linear, toy model version of Conjecture \ref{conj:instabilitych}; the nonlinear, spherically symmetric toy model version is treated in \cite{Dafermos2003a,Dafermos2005c,Dafermos2014a}. The statement (i) in Conjecture \ref{conj:instabilitych} is part of a conjectured stability statement for the subextremal Kerr exterior. See \cite{Dafermos2015d}. In this case, the linear, toy model analogue is proved in \cite{Dafermos2014c}.

Luk performed a \emph{local} construction of spacetimes with a weak null singularity and \uline{without} any symmetry assumptions in \cite{Luk2013}.

The results of the follow-up paper \cite{Gajic2015a} in the extremal Kerr interior allow us to make the following conjecture for axisymmetric perturbations of extremal Kerr characteristic initial data for (\ref{eq:ee}).
\begin{conj}
Consider axisymmetric characteristic initial data for (\ref{eq:ee}) on a hypersurface $\mathcal{H}^+$, where $\mathcal{H}^+$ has future-affine complete null generators, and a hypersurface $\uline{H}_{\textnormal{initial}}$, such that the induced geometry on both hypersurfaces is globally close to extremal Kerr data, and along $\mathcal{H}^+$ the geometry dynamically approaches that of the event horizon of extremal Kerr at a sufficiently fast polynomial rate. 

Then the maximal development can be extended beyond a non-trivial Cauchy horizon $\mathcal{CH}^+$ (emanating from $i^+$) as a Lorentzian manifold with $C^0$ metric at which the Christoffel symbols remain bounded in $L^2_{\textnormal{loc}}$, with respect to a suitable coordinate system.
\end{conj}

We do not even venture a conjecture in the case of non-axisymmetric initial data for (\ref{eq:ee}), because there are as of yet no boundedness and decay estimates for non-axisymmetric solutions $\phi$ to (\ref{eq:waveqkerr}) in the exterior, nor in the interior of extremal Kerr. Indeed, in both regions there are very significant additional obstacles that arise for non-axisymmetric $\phi$ that remain unresolved.

\subsection{Outline}
In Section \ref{sec:geom} we introduce notation and various foliations of the interior region of extremal Reissner--Nordstr\"om. In particular, we construct double-null foliations that are regular at either $\mathcal{H}^+$ or $\mathcal{CH}^+$ and discuss their properties. We state in Section \ref{sec:theorems} the main theorems that are proved in the paper. We give precise details of the requirements on the characteristic initial data that are needed for the results discussed in Section \ref{sec:discussionresults} to hold. 

In Section \ref{sec:sphericalsymmetry}, we restrict to spherically symmetric solutions $\phi=\phi_0$ of (\ref{eq:waveqkerr}). We establish $L^1$ estimates for the derivatives of $\phi_0$ along null hypersurfaces, in order to prove pointwise boundedness of both $\phi_0$ and its first-order derivatives at $\mathcal{CH}^+$, under suitable pointwise decay assumptions along $\mathcal{H}^+$ for $\phi_0$ and its tangential derivatives. Under more precise assumptions on the asymptotic behaviour of $\phi_0$, we prove moreover pointwise boundedness of second-order derivatives at $\mathcal{CH}^+$. We show that $\phi_0$ can be extended as a $C^0$, $C^1$ or $C^2$ function across $\mathcal{CH}^+$, where we need increasingly stronger assumptions on the asymptotics of $\phi_0$ along $\mathcal{H}^+$ to be able to increase the regularity at $\mathcal{CH}^+$.

In Section \ref{sec:eestimatesrn} we establish $L^2$ estimates for solutions $\phi$ to (\ref{eq:waveqkerr}) \emph{without} any symmetry restrictions, by employing suitably weighted vector fields in the null directions and performing energy estimates. We also extend these estimates to arbitrarily many derivatives of $\phi$ in the angular or ingoing null directions.

In Section \ref{sec:pointwise} we use the $L^2$ estimates of Section \ref{sec:eestimatesrn} to prove pointwise boundedness of $\phi$ and arbitrarily many tangential derivatives at $\mathcal{CH}^+$. We show moreover that $\phi$ and its derivatives in the angular or ingoing null directions can be extended as $C^0$ functions across $\mathcal{CH}^+$. Finally, by proving also pointwise estimates for the outgoing null derivative of $\phi$, we establish $C^{0,\alpha}$ extendibility of $\phi$ across $\mathcal{CH}^+$, where $\alpha<1$.

\subsection{Acknowledgments}
I wish to thank Mihalis Dafermos for suggesting the problem to me and for his guidance and invaluable advice throughout its execution. I also wish to thank Jonathan Luk and Harvey Reall for numerous helpful discussions, and Norihiro Tanahashi for sharing his numerical results.

\section{The geometry of extremal Reissner--Nordstr\"om}
\label{sec:geom}
We will discuss some basic geometric properties of extremal Reissner--Nordstr\"om and construct regular double-null coordinates that cover the regions on both sides of either the event horizon or Cauchy horizon.

\subsection{The extremal Reissner--Nordstr\"om metric}
Fix a mass parameter $M>0$. We define the \emph{exterior region of extremal Reissner--Nordstr\"om} as a manifold $\mathcal{M}_{\textnormal{ext}}$, together with a metric $g$, where $\mathcal{M}_{\textnormal{ext}}=\R\times (M,\infty)\times \s^2$ can be equipped with the coordinate chart $(t,r,\theta,\varphi)$, with $t\in \R$, $r\in (M,\infty)$, $\theta\in (0,\pi)$, $\varphi\in (0,2\pi)$. 

In these coordinates the metric $g$ in $\mathcal{M}_{\textnormal{ext}}$ is given by
\begin{equation*}
g=-\Omega^2(r)dt^2+\Omega^{-2}(r)dr^2+\sg.
\end{equation*}
Here,
\begin{equation*}
\Omega^2(r)=\left(1-\frac{M}{r}\right)^2
\end{equation*}
and $\sg=r^2 \sg_{\s^2}$, with
\begin{equation*}
\sg_{\s^2}=d\theta^2+\sin^2\theta d\varphi^2
\end{equation*}
the metric on a round sphere of radius 1.

We define the \emph{interior region of extremal Reissner--Nordstr\"om} as a manifold $\mathcal{M}_{\textnormal{int}}=\R\times (0,M)\times \s^2$ that can similarly be equipped with the chart $(t,r,\theta,\varphi)$, with $t\in \R$, $r\in (0,M)$, $\theta\in (0,\pi)$, $\varphi\in (0,2\pi)$. 

In these coordinates the metric $g$ in $\mathcal{M}_{\textnormal{int}}$ is also given by
\begin{equation*}
g=-\Omega^2(r)dt^2+\Omega^{-2}(r)dr^2+\sg.
\end{equation*}

Let $u$ and $v$ be functions in $\mathcal{M}_{\textnormal{int}}$, such that 
\begin{align*}
2v&=t+r_*,\\
2u&=t-r_*,
\end{align*}
where $r_*(r)$ is the \emph{tortoise function}, which is defined as a solution to
\begin{equation*}
\frac{dr_*}{dr}=\Omega^{-2}(r),
\end{equation*}
given explicitly by
\begin{equation}
\label{eq:rntortoise}
r_*(r)=\frac{M^2}{M-r}+2M\log(M-r)+r.
\end{equation}

We can change to \emph{ingoing} Eddington--Finkelstein coordinates $(v,r,\theta,\varphi)$ to show that the spacetime $\mathcal{M}_{\textnormal{int}}$ can be smoothly patched to the spacetime $\mathcal{M}_{\textnormal{ext}}$. The corresponding boundary of $\mathcal{M}_{\textnormal{ext}}$ inside the patched spacetime is made up of the points $(v,M,\theta,\varphi)$, with $v\in \R$, $\theta\in (0,\pi)$ and $\varphi\in (0,2\pi)$. This boundary is called the event horizon and is denoted by $\mathcal{H}^+$. It lies in the causal past of $\mathcal{M}_{\rm int}$. We denote the patched manifold by $\mathcal{M}:=\mathcal{M}_{\textnormal{int}}\cup \mathcal{M}_{\textnormal{ext}}\cup \mathcal{H}^+$.

We can also change to \emph{outgoing} Eddington--Finkelstein coordinates $(u,r,\theta,\varphi)$ in $\mathcal{M}_{\textnormal{int}}$ to show that $\mathcal{M}_{\textnormal{int}}$ can be smoothly embedded into another spacetime $\mathcal{M}'$, by patching $\mathcal{M}_{\rm int}$ to a spacetime $\mathcal{M}_{\rm ext}'$ that is isometric to $\mathcal{M}_{\rm ext}$. The corresponding boundary of $\mathcal{M}_{\rm int}$ and $\mathcal{M}'_{\rm ext}$ lies in the causal future of $\mathcal{M}_{\rm int}$ and is denoted by $\mathcal{CH}^+$. We refer to this boundary as the \emph{inner horizon}; it is composed of the points $(u,M,\theta,\varphi)$, with $u\in \R$, $\theta\in (0,\pi)$ and $\varphi\in (0,2\pi)$. We can write  $\mathcal{M}'=\mathcal{M}_{\textnormal{int}}\cup \mathcal{M}'_{\textnormal{ext}}\cup \mathcal{CH}^+$.

As $\mathcal{M}_{\rm ext}'$ is isometric to $\mathcal{M}_{\rm ext}$, the manifold $\mathcal{M}\cup \mathcal{M}'$ can be further extended to form an infinite sequence of patched manifolds isometric to either $\mathcal{M}_{\rm ext}$ or $\mathcal{M}_{\rm int}$, glued across horizons. This spacetime $\widetilde{\mathcal{M}}$ is called maximal analytically extended Reissner--Nordstr\"om and it is depicted in Figure \ref{fig:fullspacetime}. For the remainder of this paper we will, however, mainly direct our attention to the subset $\mathcal{M}\cup \mathcal{CH}^+$.

Note that $T:=\partial_t$ is a causal Killing vector field in $(\mathcal{M}\cup \mathcal{CH}^+, g)$, which is timelike in $\mathcal{M}_{\textnormal{int}}\cup \mathcal{M}_{\textnormal{ext}}$ and null along $\mathcal{H}^+$ and $\mathcal{CH}^+$. Moreover, we denote the three spacelike Killing vector fields corresponding to the spherical symmetry of the spacetime by $O_i$, with $i=1,2,3$. They can be expressed in spherical coordinates as:
\begin{align*}
O_1&=\sin \varphi \partial_{\theta}+\cot\theta \cos \varphi \partial_{\varphi},\\
O_2&=-\cos \varphi \partial_{\theta}+\cot\theta \sin \varphi \partial_{\varphi},\\
O_3&=\partial_{\varphi}.
\end{align*}

We will also consider Eddington--Finkelstein double-null coordinates, $(u,v,\theta,\varphi)$, with $u,v\in \R$, in $\mathcal{M}_{\textnormal{int}}$, in which the metric is given by
\begin{equation*}
g=-4\Omega^2dudv+\sg.
\end{equation*}
We have that $u\to -\infty$ as we approach $r=M$ along constant $v$ null hypersurfaces and $v\to \infty$ as we approach $r=M$ along constant $u$ null hypersurfaces.

\subsection{Double-null foliations of the interior region}
Since we will be considering energy estimates along both ingoing and outgoing null hypersurfaces, it is more convenient to work with double-null coordinates in $\mathcal{M}$, instead of either ingoing or outgoing Eddington--Finkelstein coordinates.

We will show that we can extend the range of the double-null Eddington-Finkelstein coordinates $(u,v,\theta,\varphi)$ from $\mathcal{M}_{\rm int}$ to $\mathcal{M}$ by a suitable rescaling of the $u$ coordinate. We can similarly extend the range from $\mathcal{M}_{\rm int}$ to $\mathcal{M}'$ by suitable rescaling of the $v$ coordinate. 

Let $-\infty<u_0<0$ and $0<v_0<\infty$. We introduce the rescaled functions $U$ and $V$, with $U\in (0,M)$ and $V\in(0,M)$, defined by
\begin{align*}
U(u)&=M-r(u,v_0),\\
V(v)&=r(u_0,v).
\end{align*}
Then we have that
\begin{align*}
\frac{dU}{du}&=-\partial_ur(u,v_0)=\Omega^2(u,v_0),\\
\frac{dV}{dv}&=\partial_vr(u_0,v)=\Omega^2(u_0,v).
\end{align*}
In double-null coordinates $(U,v,\theta,\varphi)$, the metric is given by
\begin{equation*}
g=-4\frac{\Omega^2(u,v)}{\Omega^2(u,v_0)}dUdv+\sg.
\end{equation*}

We can also express the metric in double-null coordinates $(u,V,\theta,\varphi)$ by
\begin{equation*}
g=-4\frac{\Omega^2(u,v)}{\Omega^2(u_0,v)}dudV+\sg.
\end{equation*}

We will show that we can extend $U$, a function on $\mathcal{M}_{\rm int}$, to the bigger domain $\mathcal{M}$. We can cover $\mathcal{M}$ with double-null coordinates $(\tilde{U},\tilde{V},\theta,\varphi)$, where $\tilde{U}$ is an ingoing null coordinate and $\tilde{V}$ an outgoing coordinate. That is to say, $\tilde{U}$ is constant along outgoing null hypersurfaces and $\tilde{V}$ is constant along ingoing null hypersurfaces that are preserved under isometries of spherical symmetry. We can therefore express $U=r(\tilde{U})$ and use this expression to extend $U$ as a function on the entire manifold $\mathcal{M}$. Then $U\in (0,\infty)$. Moreover, by using that $r$ is a smooth function on $\mathcal{M}$ with non-vanishing ingoing null derivative, it follows that $U=r(\tilde{U})$ must be also be a smooth function on $\mathcal{M}$. We conclude that the map
\begin{equation*}
(v,r,\theta,\varphi)\mapsto (v,U,\theta,\varphi)
\end{equation*}
is smooth. Moreover, $\partial_U^kr$ is well-defined at each point along $\mathcal{H}^+=\{U=0\}$ for all $k\in \N$.

We can use similar arguments to extend $V$ as a smooth function to the entire manifold $\mathcal{M}'$, with $V\in(0,\infty)$. The map\begin{equation*}
(u,r,\theta,\varphi)\mapsto (u,V,\theta,\varphi)
\end{equation*}
is smooth, so $\partial_V^kr$ is well-defined at each point along $\mathcal{CH}^+=\{V=M\}$ for all $k\in \N$.
\begin{figure}[h!]
\begin{center}
\includegraphics[width=2.5in]{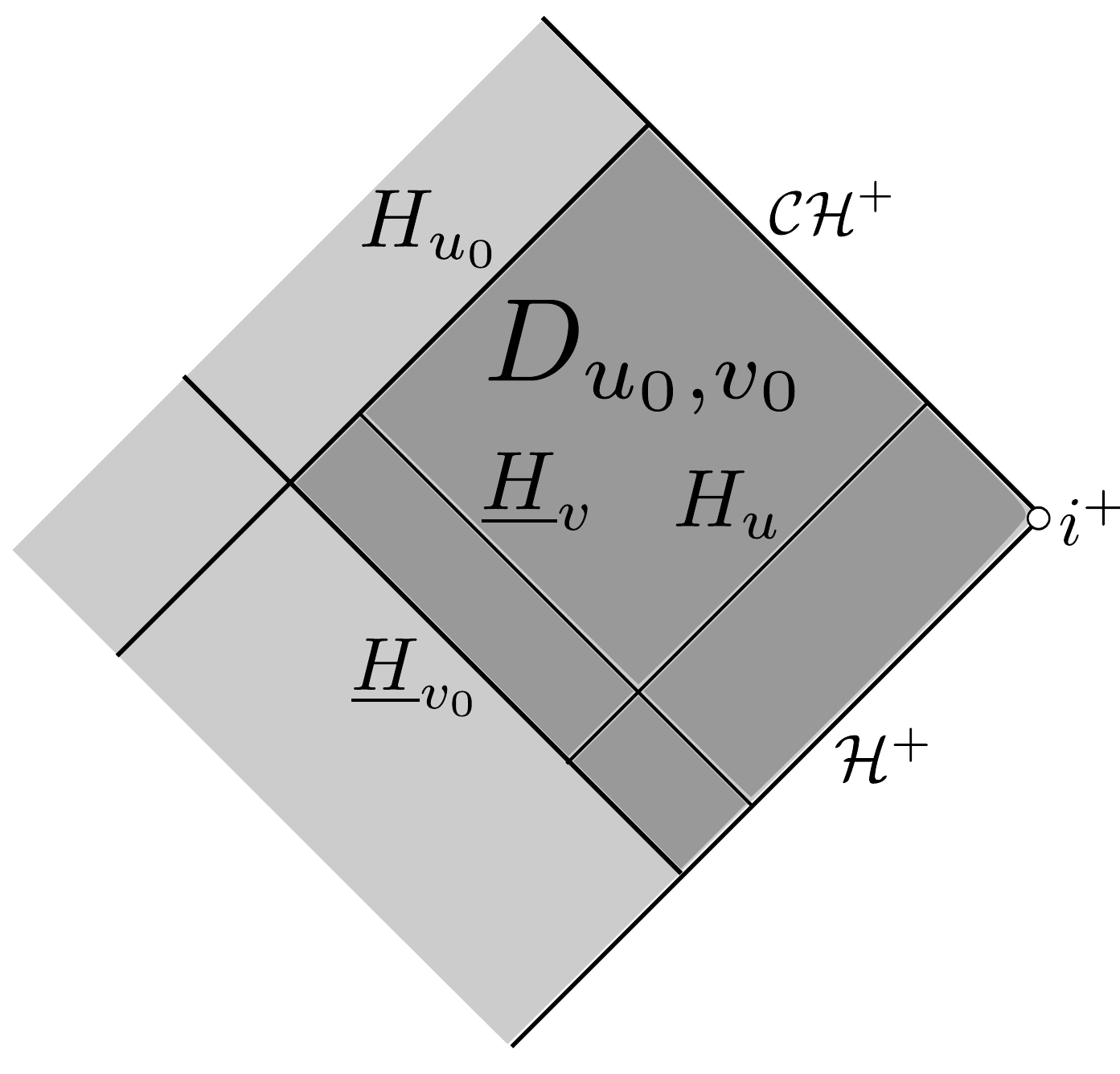}
\caption{\label{fig:ru1v1} }
\end{center}
\end{figure}

We will use the notation $(u,v,\theta,\varphi)=(-\infty,v_0,\theta,\varphi)$ and $(u,v,\theta,\varphi)=(u_0,\infty,\theta,\varphi)$, with $u_0,v_0<\infty$, for points on $\mathcal{H}^+$ and $\mathcal{CH}^+$, respectively, for the sake of convenience. These points lie in the domain of either the $(U,v)$ or $(u,V)$ double-null coordinates.

In $\mathcal{M}_{\textnormal{int}}\cup \mathcal{H}^+\cup \mathcal{CH}^+$ we restrict to the region
\begin{equation*}
D_{u_0,v_0}=\{x\in\mathcal{M}_{\textnormal{int}}\cup \mathcal{H}^+\cup \mathcal{CH}^+\::\: U(x)\in[0,U(u_0)],\:V(x)\in[V(v_0),M],\:(U(x),V(x))\neq (0,M)\}.
\end{equation*}

Let $v'\in[v_0,\infty)$ and $u'=[-\infty,u_0]$. We will often restrict to the following null hypersurfaces:
\begin{align*}
\uline{H}_{v'}&:=\{x\in \mathcal{M}\::\:U(x)\in [0,U(u_0)],\: v(x)=v'\},\\
H_{u'}&:=\{x\in \mathcal{M}\::\:U(x)=U(u'),\: v(x)\in[v_0,\infty)\}.
\end{align*}
We refer to the hypersurfaces $\uline{H}_{v'}$ and $H_{u'}$ as ingoing and outgoing null hypersurfaces, respectively.

We will now introduce some notation to group terms that decay in $v$, $|u|$, or in certain linear combinations of $v$ and $|u|$. 
\begin{definition}
\hspace{1pt}
\begin{itemize}\setlength\itemsep{1em}
\item[(i)]
Let $f: \mathcal{M}_{\rm int}\cap D_{u_0,v_0}\to \R$ be a $C^k$-function. We say that $f\in \mathcal{O}_{k}((v+|u|)^{-l})$, where $u$ and $v$ are Eddington--Finkelstein double-null coordinates in $\mathcal{M}_{\textnormal{int}}$, if there exists a constant $C=C(M,\Sigma)>0$, such that
\begin{equation*}
|\partial_u^{j_1}\partial_v^{j_2} f|(u,v,\theta,\varphi)\leq C (v+|u|)^{-l-j},
\end{equation*}
for all $0\leq j_1+j_2 \leq j$, with $j\leq k$. We also denote $\mathcal{O}((v+|u|)^{-l}):=\mathcal{O}_{0}((v+|u|)^{-l})$.
\item[(ii)]
We say $f\in \mathcal{O}_{k}(|u|^{-l})$ if there exist a constant $C=C(M,\Sigma)>0$, such that
\begin{equation*}
|\partial_u^{j} f|(u,v,\theta,\varphi)\leq C |u|^{-l-j},
\end{equation*}
for all $0\leq j \leq k$ and we use the notation $\mathcal{O}(|u|^{-l}):=\mathcal{O}_{0}(|u|^{-l})$.

Now let $f: \mathcal{M}\cap D_{u_0,v_0}\to \R$ be a $C^k$-function. Then we say that $f\in \mathcal{O}_{k}(v^{-l})$ if there exists a constant $C=C(M,\Sigma)>0$, such that
\begin{equation*}
|\partial_v^{j} f|(u,v,\theta,\varphi)\leq C v^{-l-j},
\end{equation*}
for all $0\leq j \leq k$ and we use the notation $\mathcal{O}(v^{-l}):=\mathcal{O}_{0}(v^{-l})$.
\end{itemize}
\end{definition}
We will now determine the leading-order behaviour of $\partial_Ur$ and $\partial_U^2r$ in $v$ and $|u|$.
\begin{proposition}
\label{prop:metricc1rn}
In $\mathcal{M}_{\rm int}\cap D_{u_0,v_0}$, we can expand
\begin{align*}
-\partial_Ur &=1-2\frac{v-v_0}{v+|u|}+(v-v_0)\log(v+|u|)\mathcal{O}_2((v+|u|)^{-2})\\
-\partial_u\partial_Ur&=-\Omega^{2}(v_0,u)\partial_U^2r =2\frac{v-v_0}{(v+|u|)^2}+(v-v_0)\log(v+|u|)\mathcal{O}_1((v+|u|)^{-3})
\end{align*}\end{proposition}
\begin{proof}
We use expression (\ref{eq:rntortoise}) to obtain the following implicit relation between $r$, $u$ and $v$:
\begin{equation*}
\frac{r}{M-r}=\frac{M}{M-r}-1=M^{-1}(v-u)-2\log(M-r)+M^{-1}(M-r)-2.
\end{equation*}
Consequently, we can express
\begin{equation}
\label{eq:rnOmegainverseexp}
\begin{split}
\Omega^{-2}(u,v)&=\frac{r^2}{(M-r)^2}=M^{-2}(v-u)^2-4M^{-1}(v-u)\log(M-r)(u,v)-4M^{-1}(v-u)\\
&\quad+4(\log(M-r)(u,v))^2+8\log(M-r)(u,v)+\mathcal{O}_{2}(1).
\end{split}
\end{equation}

We can estimate
\begin{equation}
\label{eq:exppartialUR}
\begin{split}
-\partial_Ur(u,v)&=\frac{\Omega^{-2}(u,v_0)}{\Omega^{-2}(u,v)}\\
&=\left(1-\frac{v-v_0+2M\log\left(\frac{(M-r)(u,v)}{(M-r)(u,v_0)}\right)}{v+|u|-2M\log (M-r)(u,v)-2+(M-r)(u,v)}\right)^2\\
&\qquad+\mathcal{O}_2((v+|u|)^{-2})\\
&=\left(1-\frac{v-v_0}{v+|u|}+\frac{2M\log\left(\frac{(M-r)(u,v)}{(M-r)(u,v_0)}\right)}{v+|u|-2M\log (M-r)(u,v)-2+(M-r)(u,v)}\right)^2\\
&\qquad+\log(v+|u|)(v-v_0)\mathcal{O}_2((v+|u|)^{-2})\\
&=1-2\frac{v-v_0}{v+|u|}+\log(v+|u|)(v-v_0)\mathcal{O}_2((v+|u|)^{-2}).
\end{split}
\end{equation}

In particular, it follows immediately that
\begin{equation*}
-\partial_Ur(0,v)=\lim_{U\downarrow 0} (-\partial_U r)(U,v)=\lim_{u\to -\infty} \frac{\Omega^{-2}(u,v_0)}{\Omega^{-2}(u,v)}=1.
\end{equation*}

Furthermore, we can differentiate (\ref{eq:exppartialUR}) in $u$ to obtain
\begin{equation*}
\begin{split}
\Omega^2(u,v_0)\partial_U^2r&=\partial_u\partial_U r=2\frac{v-v_0}{(v+|u|)^2}+\log(v+|u|)(v-v_0)\mathcal{O}_1((v+|u|)^{-3})
\end{split}
\end{equation*}

In particular, it follows that
\begin{equation*}
\partial_U^2r(0,v)=\lim_{U\downarrow 0}\partial_U^2r=2M^{-2}(v-v_0). \qedhere
\end{equation*}
\end{proof}  

We will also need a precise expansion in $u$ and $v$ of $\partial_V^2r$ in $\mathcal{M}_{\textnormal{int}}$.
\begin{proposition}
\label{prop:part2dervr}
In $\mathcal{M}_{\rm int}\cap D_{u_0,v_0}$, we can expand
\begin{align}
\label{eq:part1dervr}
\partial_Vr &=1-2\frac{|u|-|u_0|}{v+|u|}+(|u|-|u_0|)\log(v+|u|)\mathcal{O}_2((v+|u|)^{-2}),\\
\label{eq:part2dervr}
\partial_v\partial_Vr&=\Omega^{2}(u_0,v)\partial_V^2r =2\frac{|u|-|u_0|}{(v+|u|)^2}+(|u|-|u_0|)\log(v+|u|)\mathcal{O}_1((v+|u|)^{-3}).
\end{align}
\end{proposition}
\begin{proof}
We repeat the arguments in Proposition \ref{prop:metricc1rn}, with $v$ and $|u|$ interchanged.
\end{proof}

\subsection{The divergence theorem and integration norms}
\label{sec:vectorfieldmethod}
In this section we will introduce some basic notation corresponding to the ``vector field method'', which was mentioned in Section \ref{sec:mainideasThm12}. In particular, we will set notation for spacetime integrals and integrals over null hypersurfaces.

Let $V$ be a vector field in a Lorentzian manifold $(\mathcal{N},g)$. We consider the stress-energy tensor $\mathbf{T}[\phi]$ corresponding to (\ref{eq:waveqkerr}), with components
\begin{equation*}
\mathbf{T}_{\alpha \beta}[\phi]=\partial_{\alpha}\phi \partial_{\beta} \phi-\frac{1}{2}g_{\alpha \beta} \partial^{\gamma} \phi \partial_{\gamma} \phi.
\end{equation*}

Let $J^V[\phi]$ denote the energy current corresponding to $V$, which is obtained by applying $V$ as a vector field multiplier, i.e.\ in components
\begin{equation*}
J^V_{\alpha}[\phi]=\mathbf{T}_{\alpha \beta}[\phi] V^{\beta}.
\end{equation*}

An energy flux is an integral of $J^V[\phi]$ contracted with the normal to a hypersurface with the natural volume form corresponding to the metric induced on the hypersurface. We apply the divergence theorem to relate the energy flux along the boundary of a spacetime region to the spacetime integral of the divergence of the energy current $J^V$. If the boundary has a null segment, there is no natural volume form, so the volume form is chosen in such a way that the divergence theorem holds.

That is to say, if we take $-\infty\leq u_1<u_2\leq u_0$ and $v_0\leq v_1<v_2\leq \infty$, the divergence theorem in the open rectangle $\{u_1<u<u_2,\:v_1<v<v_2\}$ in $\mathcal{M}_{\rm int}$ gives the following identity:
\begin{equation}
\label{divergenceidentity}
\begin{split}
\int_{\{u_1<u<u_2,\:v_1<v<v_2\}} \textnormal{div}J^V[\phi]&=-\int_{H_{u_2}\cap \{v_1\leq v\leq v_2\}}J^V[\phi]\cdot \partial_v+\int_{H_{u_1}\cap \{v_1\leq v\leq v_2\}}J^V[\phi]\cdot \partial_v\\
&\quad - \int_{\uline{H}_{v_2}\cap \{u_1\leq u\leq u_2\}}J^V[\phi]\cdot \partial_u+\int_{\uline{H}_{v_1}\cap \{u_1\leq u\leq u_2\}}J^V[\phi]\cdot \partial_u.
\end{split}
\end{equation}

Here, we introduced the following notation:
\begin{equation*}
J^V[\phi]\cdot W=\mathbf{T}(V,W),
\end{equation*}
for vector fields $V$ and $W$. Moreover, in the notation on the left-hand side of (\ref{divergenceidentity}), we integrate over spacetime with respect to the standard volume form, i.e.\ let $f: \mathcal{M}\cap D_{u_0,v_0}\to \R$ be a suitably regular function and $U$ an open subset of $\mathcal{M}$, then
\begin{equation*}
\int_{U}f:=\int_{U}f(u,v,\theta,\varphi)\,\sqrt{-\det g}dudvd\theta d\varphi= \int_{U}f(u,v,\theta,\varphi)\,2\Omega^2r^2d\mu_{\s^2}dudv,
\end{equation*}
where $d\mu_{\s^2}$ is the standard volume form on the round sphere  $\s^2$ of radius 1.

When integrating over $H_u$ and $\uline{H}_{v}$, we do not have a standard volume form at our disposal, so we used the following convention in the notation on the right-hand side of (\ref{divergenceidentity}):
\begin{align*}
\int_{\uline{H}_v}f&:= \int_{0}^{U(u_0)}\int_{\s^2}f\,r^2d\mu_{\s^2}dU=\int_{-\infty}^{u_0}\int_{\s^2}f\,r^2d\mu_{\s^2}du,\\
\int_{{H}_u}f&:=\int_{V(v_0)}^M\int_{\s^2}f\,r^2d\mu_{\s^2}dV= \int_{v_0}^{\infty}\int_{\s^2}f\,r^2 d\mu_{\s^2}dv.
\end{align*}

In the notation of \cite{Christodoulou2008} we decompose the divergence term appearing in (\ref{divergenceidentity}) in the following way:
\begin{equation*}
\textnormal{div} J^V[\phi]=K^V[\phi]+\mathcal{E}^V[\phi],
\end{equation*}
where
\begin{align*}
K^V[\phi]&:=\mathbf{T}^{\alpha \beta}[\phi]\nabla_{\alpha}V_{\beta},\\
\mathcal{E}^V[\phi]&:=V(\phi)\square_g\phi.
\end{align*}
In particular, $\mathcal{E}^V[\phi]=0$ if $\phi$ is a solution to (\ref{eq:waveqkerr}).

We can write
\begin{align*}
\int_{\uline{H}_v}(\partial_Uf)^2&= \int_{0}^{U(u_0)}\int_{\s^2}(\partial_Uf)^2\,r^2d\mu_{\s^2}dU= \int_{-\infty}^{u_0}\int_{\s^2}\frac{du}{dU}(\partial_uf)^2\,r^2 d\mu_{\s^2}du,\\
\int_{{H}_u}(\partial_Vf)^2&=\int_{V(v_0)}^M\int_{\s^2}(\partial_Vf)^2\,r^2d\mu_{\s^2}dV= \int_{v_0}^{\infty}\int_{\s^2}\frac{dv}{dV}(\partial_vf)^2\,r^2 d\mu_{\s^2}dv.
\end{align*}

We have that
\begin{align*}
\frac{du}{dU}&=-\frac{1}{\partial_ur (u,v_0)},\\
\frac{dv}{dV}&=\frac{1}{\partial_vr (u_0,v)}.
\end{align*}

By (\ref{eq:rnOmegainverseexp}), there exist constants $C_1=C_1(M,u_0,v_0)>0$ and $C_2=C_2(M,u_0,v_0)>0$ such that
\begin{align*}
C_1 u^2\leq \frac{du}{dU}\leq C_2 u^2,\\
C_1 v^2\leq\frac{dv}{dV}\leq C_2 v^2.
\end{align*}

We rewrite the estimates above by using the following notation:
\begin{align}
\label{est:utoU}
\frac{du}{dU}\sim u^2,\\
\label{est:vtoV}
\frac{dv}{dV}\sim v^2,
\end{align}
so that
\begin{align*}
\int_{\uline{H}_v}(\partial_Uf)^2 &\sim \int^{u_0}_{-\infty}\int_{\s^2}u^2(\partial_uf)^2\,d\mu_{\s^2}du,\\
\int_{{H}_u}(\partial_Vf)^2 & \sim \int_{v_0}^{\infty}\int_{\s^2}v^2(\partial_vf)^2\,d\mu_{\s^2}dv.
\end{align*}

Now consider a compact subset $K\subset \mathcal{M}'$, such that moreover $K\subset D_{u_0,v_0}$. Then we define the following $L^2$ norms:
\begin{align*}
||f||^2_{L^2(K)}&:=\int_{K\cap \mathcal{M}_{\rm int}} f^2,\\
||\partial f||^2_{L^2(K)}&:=\int_{K\cap \mathcal{M}_{\rm int}} (\partial_Vf)^2+(\partial_U f)^2+|\snabla f|^2,\\
\end{align*}
where $\snabla$ denotes the covariant derivative restricted to the spheres that are preserved under isometries corresponding to spherical symmetry.

We can in particular estimate
\begin{equation}
\label{est:energytoH1}
\begin{split}
||\partial f||_{L^2(K)}^2&=\int_{K\cap \mathcal{M}_{\rm int}}(\partial_Vf)^2+(\partial_U f)^2+|\snabla f|^2\leq C\sup_{v_0\leq v<\infty} \int_{\uline{H}_v}(\partial_U f)^2+\Omega^2 |\snabla f|^2\\
&\quad+ C\sup_{-\infty\leq u<u_0}\int_{{H}_u}(\partial_Vf)^2,
\end{split}
\end{equation}
where $C=C(K)>0$.

Consider now the following weighted vector field $N$ in Eddington--Finkelstein double-null coordinates:
\begin{equation*}
N=N^u(u,v)\partial_u+N^v(u,v)\partial_v.
\end{equation*}
The corresponding compatible current $K^N$ is given by
\begin{equation*}
K^N[\phi]:=\mathbf{T}_{\alpha\beta}[\phi]\prescript{N}{}\pi^{\alpha\beta},
\end{equation*}
with the components of the deformation tensor $\prescript{N}{}\pi_{\alpha\beta}=\frac{1}{2}[\nabla_{\alpha}N_{\beta}+\nabla_{\beta}N_{\alpha}]$ given by
\begin{align*}
\prescript{N}{}\pi^{uu}=g^{uv}g^{uv}\prescript{N}{}\pi_{vv}&=\frac{1}{4}\Omega^{-4}g(\nabla_{\partial_u} N,\partial_u)=-\frac{1}{2}\Omega^{-2}\partial_uN^v,\\
\prescript{N}{}\pi^{vv}=g^{uv}g^{uv}\prescript{N}{}\pi_{uu}&=\frac{1}{4}\Omega^{-4}g(\nabla_{\partial_v} N,\partial_v)=-\frac{1}{2}\Omega^{-2}\partial_vN^u,\\
\prescript{N}{}\pi^{uv}&=g^{uv}g^{uv}\prescript{N}{}\pi_{uv}=\frac{1}{8}\Omega^{-4}\left[g(\nabla_{\partial_v} N,\partial_u)+g(\nabla_{\partial_u} N,\partial_v)\right]\\
&=-\frac{1}{4}\Omega^{-2}\left(\partial_uN^u+\partial_vN^v+\Omega^{-2}N^{\sigma}\partial_{\sigma}\Omega^2\right),\\
\prescript{N}{}\pi^{uA}&=\prescript{N}{}\pi^{vA}=0,\\
\prescript{N}{}\pi^{AB}&=\frac{1}{2}\slashed{g}^{AC}\slashed{g}^{BD}(N^{u}\partial_{u}+N^v\partial_v)\slashed{g}_{CD}=r^{-1}\Omega^2(N^v-N^u)\slashed{g}^{AB}
\end{align*}

Moreover,
\begin{align*}
\mathbf{T}_{vv}[\phi]&=(\partial_v\phi)^2,\\
\mathbf{T}_{uu}[\phi]&=(\partial_u\phi)^2\\
\mathbf{T}_{uv}[\phi]&=\Omega^2|\snabla\phi|^2,\\
\mathbf{T}_{AB}[\phi]&=(\partial_A\phi)(\partial_B\phi)+\frac{1}{2}\slashed{g}_{AB}(\Omega^{-2}\partial_u\phi\partial_v\phi-|\snabla\phi|^2).
\end{align*}
Consequently,
\begin{equation}
\label{eq:genexprKN}
K^N[\phi]=r^{-1}(N^v-N^u)\partial_u\phi\partial_v\phi-\frac{1}{2}\left[\partial_vN^v+\partial_uN^u+\Omega^{-2}\partial_u(\Omega^2)N^u+\Omega^{-2}\partial_v(\Omega^2)N^v\right]|\snabla \phi|^2.
\end{equation}

\section{Main results}
\label{sec:theorems}
In this section we present precise statements of the main results proved in this paper. We first give a formulation of the standard global existence and uniqueness for the characteristic initial value problem for (\ref{eq:waveqkerr}) in $\mathcal{M}\cap D_{u_0,v_0}$.

\begin{proposition}
\label{prop:wellposedness}
Let $\phi$ be a continuous function on the union of null hypersurfaces 
\begin{equation*}
\left(\mathcal{H}\cap\{v\geq v_0\}\right)\cup \uline{H}_{v_0},
\end{equation*}
such that the restriction to $\mathcal{H}^+$ and the restriction to $\uline{H}_{v_0}$ are smooth functions. Then there exists a unique, smooth extension of $\phi$ to $\mathcal{M}_{\rm int}\cap D_{u_0,v_0}$ that satisfies (\ref{eq:waveqkerr}). We also denote this extension by $\phi$. We refer to the restriction $\phi|_{\left(\mathcal{H}\cap\{v\geq v_0\}\right)\cup \uline{H}_{v_0}}$ as characteristic initial data.
\end{proposition}

It is important to note that Proposition \ref{prop:wellposedness} does not give any information about the asymptotic behaviour of $\phi$ towards $\mathcal{CH}^+$. We will state in the subsections below further properties of $\phi$, relating to boundedness and extendibility of $\phi$ beyond $\mathcal{CH}^+$. We will need to impose suitable additional decay requirements along $\mathcal{H}^+$.

\subsection{Energy estimates along null hypersurfaces}
Consider solutions $\phi$ to (\ref{eq:waveqkerr}) that arise from the characteristic initial data in Proposition \ref{prop:wellposedness}. 

We will also consider higher-order derivatives of $\phi$. Let us introduce the differential operator $X_{m,n,l}$ that is defined as follows:
\begin{equation*}
X_{m,n,l}\phi:=\partial_u^m\partial_v^nO^l\phi.
\end{equation*}

We can prove boundedness of weighted $L^2$ norms along null hypersurfaces for $X_{m,n,l}\phi$, under additional assumptions on suitable initial $L^2$ norms.
\begin{theorem}[Higher-order weighted energy estimates]
\label{thm:eestimatesint}
Let $m,n\geq 0$ and take $0<q\leq 2$. Let $\phi$ be a solution to (\ref{eq:waveqkerr}) corresponding to initial data from Proposition \ref{prop:wellposedness} satisfying
\begin{equation*}
\begin{split}
E_{q;(m,n,l)}[\phi]&:=\sum_{m'=\min\{1,m\}}^m\sum_{n'=\min\{1,n\}}^n\sum_{s=0}^{\max\{m-m',n-n'\}}\Bigg[ \int_{\mathcal{H}^+\cap\{v\geq v_0\}} v^q(\partial_vX_{m',n',l+s}\phi)^2\\
&\quad+|\snabla X_{m',n',l+s}\phi|^2+ \int_{\uline{H}_{v_0}} |u|^2(\partial_uX_{m',n',l+s}\phi)^2+\Omega^2|\snabla X_{m',n',l+s}\phi|^2\Bigg]<\infty.
\end{split}
\end{equation*}

Then there exists a constant $C=C(M,u_0,v_0,q,m,n)>0$ such that for all $H_u$ and $\uline{H}_v$,
\begin{equation*}
\begin{split}
&\int_{H_u} v^q(\partial_v X_{m,n,l}\phi)^2+|u|^2\Omega^2|\snabla X_{m,n,l}\phi|^2+\int_{\uline{H}_v} |u|^2(\partial_uX_{m,n,l}\phi)^2+v^q\Omega^2|\snabla X_{m,n,l}\phi|^2\\
&\leq CE_{q;(m,n,l)}[\phi].
\end{split}
\end{equation*}
\end{theorem}

Theorem \ref{thm:eestimatesint} is proved in Proposition \ref{prop:mainenergyestimate} and Proposition \ref{prop:highordeestimate}. Using the estimate (\ref{est:energytoH1}) it follows that Theorem \ref{thm:eestimatesint} with $q=2$ implies Theorem \ref{thm:H1boundv1}.

\subsection{Pointwise estimates and continuous extendibility beyond $\mathcal{CH}^+$}
We can also obtain $L^{\infty}$ estimates in $\mathcal{M}\cap D_{u_0,v_0}$ and moreover extend $\phi$ as a continuous function beyond $\mathcal{CH}^+$.
\begin{theorem}[$L^{\infty}$ estimate $\partial_u^m\partial_v^nO^l\phi$]
\label{thm:pointwiseestho}
Let $m,n\geq 0$ and take $0<q\leq 2$. Let $\phi$ be a solution to (\ref{eq:waveqkerr}) corresponding to initial data from Proposition \ref{prop:wellposedness} satisfying $\sum_{|k|\leq 2}E_{q;(m,n,l)}[O^k\phi]<\infty$, such that moreover 
\begin{equation*}
\sum_{|k|\leq 2} \sup_{v_0\leq {v'} \leq \infty}\int_{\s^2} |\partial_u^m\partial_v^n O^{l+k} \phi|^2\,d\mu_{\s^2}\Big|_{U=0}<\infty.
\end{equation*}

Then there exists a constant $C(M,u_0,v_0,m,n,q)>0$ such that
\begin{equation*}
\begin{split}
||X_{m,n,l} \phi||_{L^{\infty}(H_u)}^2 &\leq C\sum_{|k|\leq 2} \sup_{v_0\leq {v'} \leq \infty}\int_{\s^2} |\partial_u^m\partial_v^n O^{l+k} \phi|^2\,d\mu_{\s^2}\Big|_{U=0}+ C|u|^{-1}\sum_{|k|\leq 2}E_{q;(m,n,l)}[O^k\phi].
\end{split}
\end{equation*}
\end{theorem}
Theorem \ref{thm:pointwiseestho} is proved in Proposition \ref{prop:boundpsirn}.

\begin{theorem}[$C^0$-extendibility of $\partial_u^m\partial_v^nO^l\phi$]
\label{thm:C0extension}
Let $m,n\geq 0$. Let $\phi$ be a solution to (\ref{eq:waveqkerr}) corresponding to initial data from Proposition \ref{prop:wellposedness} satisfying $\sum_{|k|\leq 2}E_{q;(m,n,l)}[O^k\phi]<\infty$ for $q>1$. Then $X_{m,n,l}\phi$ can be extended as a $C^0$ function beyond $\mathcal{CH}^+\cap D_{u_0,v_0}$.
\end{theorem}
Theorem \ref{thm:C0extension} is proved in Proposition \ref{prop:C0extension}. Moreover Theorem \ref{thm:pointwiseestho} and \ref{thm:C0extension}, together with the estimates along the event horizon in Theorem 4 of \cite{Aretakis2011}, imply Theorem \ref{thm:linftyboundv1}.

\begin{theorem}[$L^{\infty}$ estimates for $\partial_V\phi$]
\label{thm:estpartialVphi}
Let $\phi$ be a solution to (\ref{eq:waveqkerr}) corresponding to initial data from Proposition \ref{prop:wellposedness} and denote
\begin{equation*}
D=||\partial_U\phi||^2_{L^{\infty}\left(\uline{H}_{v_0}\right)}+||\snabla\phi ||^2_{L^{\infty}\left(\uline{H}_{v_0}\right)}.
\end{equation*}
\hspace{1pt}
\begin{itemize}\setlength\itemsep{1em}
\item[(i)]
Let $0<q\leq 2$ and take $\epsilon>0$ to be arbitrarily small. Assume that $\sum_{|k|\leq 2}E_{q}[O^k\phi]<\infty$, where $E_{q}[\phi]:=E_{q;(0,0,0)}[\phi]$, and moreover
\begin{equation*}
\begin{split}
\tilde{E}_{1;\epsilon}[\phi]&:=\int_{\mathcal{H}^+\cap\{v\geq v_0\}} v^{1+\epsilon}\left[({\partial_v}\phi)^2+|\snabla\phi|^2+|\snabla^2\phi|^2\right]+ D<\infty.
\end{split}
\end{equation*}
Then there exists a constant $C=C(M,v_0,u_0,\epsilon,q)>0$ such that
\begin{equation*}
\begin{split}
||\partial_V\phi||^2_{L^{\infty}(S^2_{u,v})}&\leq C\sum_{|k|\leq 2}\int_{S^2_{-\infty,v}}v^{2}({\partial_v}O^k\phi)^2\,d\mu_{\slashed{g}_{\s^2}}+C|u|^{-q}\sum_{|k|\leq 2}E_{q}[O^k\phi]\\
&\quad+C\left(\frac{v}{v+|u|}\right)^4\log\left(\frac{v}{v+|u|}\right)\sum_{1\leq |k|\leq 4}\tilde{E}_{1;\epsilon}[O^k\phi]\\
&\quad+C(v+|u|)^{-2}\sum_{1\leq |k|\leq 4} \int_{S^2_{-\infty,v}}(O^k\phi)^2\,d\mu_{\s^2}.
\end{split}
\end{equation*}
In particular, $\phi$ can be extended as a $C^{0,\alpha}$ function beyond $\mathcal{CH}^+$, for any $\alpha<1$.
\item[(ii)]
Let $\phi=\phi_0$, the spherically symmetric angular mode. Then we can estimate
\begin{equation*}
||\partial_V\phi_0||^2_{L^{\infty}(S^2_{u,v})}\leq Cv^{2}({\partial_v}\phi_0)^2(-\infty,v)+C|u|^{-q}E_{q}[\phi_0].
\end{equation*}
\end{itemize}
\end{theorem}

Theorem \ref{thm:estpartialVphi} is proved in Proposition \ref{prop:boundderVpsi} and implies Theorem \ref{thm:c11boundphi}.

\subsection{$C^1$- and $C^2$-extendibility of spherically symmetric solutions}
If we restrict to spherically symmetric initial data in Proposition \ref{prop:wellposedness}, we can further show extendibility in $C^1$ and $C^2$, under suitable decay assumptions along $\mathcal{CH}^+$.
\begin{theorem}[$L^{\infty}$ estimate for $\partial_V^2\phi_0$]
\label{thm:bound2nddersphsymm}
Let $\phi_0$ be a solution to (\ref{eq:waveqkerr}) corresponding to spherically symmetric initial data from Proposition \ref{prop:wellposedness}, which also satisfies the following asymptotics along $\mathcal{H}^+$:
\begin{equation*}
\phi_0|_{\mathcal{H}^+}(v)=-M^2H_0v^{-1}+M^3H_0v^{-2}\log v +\mathcal{O}_2(v^{-2}).
\end{equation*} 

Then there exists a constant $C_0>0$, depending on the initial data for $\phi_0$, such that
\begin{equation*}
|\partial_V^2\phi_0|(u,v)\leq C_0.
\end{equation*}
for all $-\infty <u \leq u_0$ and $v_0\leq v < \infty$.
\end{theorem}

Theorem \ref{thm:bound2nddersphsymm} follows from Proposition \ref{cor:c2boundphi0}.

\begin{theorem}[$C^1$- and $C^2$-extendibility of $\phi_0$]
\label{thm:c1c2ext}
Let $\phi_0$ be a solution to (\ref{eq:waveqkerr}) corresponding to spherically symmetric initial data from Proposition \ref{prop:wellposedness}.
\hspace{1pt}
\begin{itemize}\setlength\itemsep{1em}
\item[(i)]
Assume that
\begin{equation*}
\lim_{v\to \infty} v^2 \partial_v\phi_0|_{\mathcal{H}^+}(v) \: \textnormal{is well-defined}.
\end{equation*}
Then $\phi_0$ can be extended as a $C^1$ function beyond $\mathcal{CH}^+$.
\item[(ii)]
Assume that $\phi_0$ satisfies the following asymptotics along $\mathcal{H}^+$:
\begin{equation*}
\phi_0|_{\mathcal{H}^+}(v)=-M^2H_0v^{-1}+M^3H_0v^{-2}\log v +\mathcal{O}_2(v^{-2}),
\end{equation*}
such that moreover
\begin{equation}
\label{eq:addassumpinitial2}
\lim_{v\to \infty}\left[M\partial_V^2\phi_0(-\infty,v)+2H_0\log v\right]\:\textnormal{is well-defined}.
\end{equation}
Then $\phi_0$ can be extended as a $C^2$ function beyond $\mathcal{CH}^+$.
\end{itemize}
\end{theorem}
Theorem \ref{thm:c1c2ext} (i) follows from Proposition \ref{prop:sphsymmc1} and Theorem \ref{thm:c1c2ext} (ii) follows from Proposition \ref{prop:sphsymmc1} and \ref{prop:sphsymmc2}, together with Theorem \ref{thm:C0extension}. Theorem \ref{thm:c1c2ext} (i) and (ii) imply Theorem \ref{thm:c1boundphi0} and Theorem \ref{thm:c2blowup}, respectively.\footnote{Note that the assumption (\ref{eq:addassumpinitial2}) in Theorem \ref{thm:c1c2ext} (ii) can easily be seen to be equivalent to the assumption (\ref{eq:addassumpinitial}) in Theorem \ref{thm:c2blowup}, by using the precise $v$ dependence of the expression $\frac{dV}{dv}$ from Section \ref{sec:geom}.}

\section{The spherically symmetric mode $\phi_0$}
\label{sec:sphericalsymmetry}
We first consider the special case of spherically symmetric solutions $\phi_0$ to the wave equation (\ref{eq:waveqkerr}). \textbf{In this section $\phi_0$ will always denote a solution to} (\ref{eq:waveqkerr}) \textbf{corresponding to characteristic initial data from Proposition} \ref{prop:wellposedness} \textbf{that we additionally assume to be spherically symmetric}. It follows that the solution $\phi_0$ must be spherically symmetric in the entire set $\mathcal{M}\cap D_{u_0,v_0}$.

In the spherically symmetric setting we can do estimates with respect to weighted $L^1$ norms. We can use these to prove uniform boundedness of $\phi_0$ and $\partial_V\phi_0$. We can moreover show that $\partial_V^2\phi_0$ is uniformly bounded, if we assume precise asymptotics for $\phi$ along $\mathcal{H}^+$, and that $\partial_V^2\phi_0$ blows up at $\mathcal{CH}^+$ if these asymptotics do not exactly hold. 

\subsection{Weighted $L^1$ estimates}
\label{sec:weightedL1}
Uniform boundedness of spherically symmetric solutions $\phi_0$ and their derivatives follows from considering appropriately weighted $L^1$ norms for the derivatives of $\phi_0$. We can rewrite the wave equation (\ref{eq:waveqkerr}) as a system of transport equations for ${r\partial_u \phi_0}$ and ${r\partial_v \phi_0}$:
\begin{align}
\label{eq:weth}
\partial_u ({r\partial_v \phi_0})&=-\partial_v r{\partial_u \phi_0},\\
\label{eq:wez}
\partial_v ({r\partial_u \phi_0})&=-\partial_u r{\partial_v \phi_0}.
\end{align}
These transport equations allow us to obtain weighted $L^1$ estimates for $r\partial_v \phi_0$ and ${r\partial_u \phi_0}$ that will be central in proving pointwise boundedness results.

We first prove a lemma that can be interpreted as Gr\"onwall's inequality in two variables.
\begin{lemma}
\label{lm:gronwall}
Let $-\infty \leq u_1<u_2\leq \infty$ and $-\infty \leq v_1<v_2\leq \infty$. Consider continuous, non-negative functions $f,g: [u_1,u_2]\times [v_1,v_2] \to \R$, and continuous, non-negative functions $h: [u_1,u_2]\to \R$ and $k: [v_1,v_2]\to \R$. Suppose
\begin{equation}
\label{eq:gronassumption}
f(u,v)+g(u,v)\leq A+ B\left[\int_{u_1}^u h(u')f(u',v)\,du'+\int_{v_1}^v k(v')g(u,v')\,dv'\right],
\end{equation}
for all $u\in [u_1,u_2]$ and $v\in [v_1,v_2]$, where $A,B>0$ are constants. Then:
\begin{equation}
\label{eq:gronba}
f(u,v)+g(u,v)\leq (1+\eta)A e^{\beta B\left[\int_{u_1}^u h(u')\,du'+\int_{v_1}^v k(v')\,dv'\right]},
\end{equation}
for all $u\in [u_1,u_2]$ and $v\in [v_1,v_2]$, where $\eta>0$ can be taken arbitrarily small and $\beta\geq\frac{2(1+\eta)}{\eta}$.
\end{lemma}
\begin{proof}
We will prove the lemma using a continuity argument. Consider the set
\begin{equation*}
S:=\{(u,v)\in [u_1,u_2]\times [v_1,v_2]\,:\, \textnormal{the inequality}\, (\ref{eq:gronba})\, \textnormal{holds}\}.
\end{equation*}
Since (\ref{eq:gronba}) trivially holds for all points $(u_1,v)$, with $v\in [v_1,v_2]$ and $(u,v_1)$, with $u\in [u_1,u_2]$, by (\ref{eq:gronassumption}), $S$ is nonempty. Moreover, by continuity of $f$ and $g$, $S$ is closed. We are left with showing that $S$ is open. It is sufficient to show openess via a \emph{bootstrap argument}, i.e.\ we will show that for all $(u,v)\in S$:
\begin{equation}
\label{eq:gronimprov}
f(u,v)+g(u,v)\leq \alpha A e^{\beta B\left[\int_{u_1}^u h(u')\,du'+\int_{v_1}^v k(v')\,dv'\right]},
\end{equation}
for some $\alpha<1+\eta$. Indeed, if we can show (\ref{eq:gronimprov}) for $(u,v)\in S$, $S$ must also contain an open neighbourhood of $(u,v)$, by continuity of $f$ and $g$.

By combining (\ref{eq:gronassumption}) and (\ref{eq:gronba}) we obtain
\begin{equation*}
\begin{split}
&f(u,v)+g(u,v)\\
&\leq A+ (1+\eta)A B\left[e^{\beta B\int_{v_1}^v k(v')\,dv'}\int_{u_1}^u h(u') e^{\beta B\int_{u_1}^{u'} h(\tilde{u})\,d\tilde{u}}\,du'+e^{\beta B\int_{u_1}^u h(u')\,du'}\int_{v_1}^v k(v')e^{ \beta B\int_{v_1}^{v'} k(\tilde{v})\,d\tilde{v}}\,dv'\right]\\
&= A+ \frac{1+\eta}{\beta}A \Bigg[e^{\beta B\int_{v_1}^v k(v')\,dv'}\int_{u_1}^u \frac{d}{du'}\left(e^{\beta B\int_{u_1}^{u'} h(\tilde{u})\,d\tilde{u}}\right)\,du'\\
&\quad+e^{\beta B\int_{u_1}^u h(u')\,du'}\int_{v_1}^v \frac{d}{dv'}\left(e^{ \beta B\int_{v_1}^{v'} k(\tilde{v})\,d\tilde{v}}\right)\,dv'\Bigg]\\
&\leq A+\frac{1+\eta}{\beta}Ae^{\beta B\int_{v_1}^v k(v')\,dv'}\left(e^{\beta B\int_{u_1}^u h(u')\,du'}-1\right)+\frac{1+\eta}{\beta}Ae^{\beta B\int_{u_1}^u h(u')\,du'}\left(e^{\beta B\int_{v_1}^v k(v')\,dv'}-1\right)\\
&\leq A\left(1+\frac{2(1+\eta)}{\beta}e^{\beta B\left[\int_{u_1}^u h(u')\,du'+\int_{v_1}^v k(v')\,dv'\right]}\right)\\
&\leq \left(1+\frac{2(1+\eta)}{\beta}\right)Ae^{\beta B\left[\int_{u_1}^u h(u')\,du'+\int_{v_1}^v k(v')\,dv'\right]}.
\end{split}
\end{equation*}
Hence, (\ref{eq:gronimprov}) holds if $\beta>\frac{2(1+\eta)}{\eta}$.

We conclude that $S$ is non-empty, open and closed, so by connectedness of $[u_1,u_2]\times [v_1,v_2]$, $S=[u_1,u_2]\times [v_1,v_2]$. By continuity, we can take $\beta\geq \frac{2(1+\eta)}{\eta}$.
\end{proof}

\begin{proposition}
\label{lm:l1}
Let $-\infty<u<u_0<0$ and $v_0<v<\infty$. For $-1< q<1$ and $0\leq p< 1-q$, there exists a constant $C=C(M,u_0,v_0,p,q)>0$ such that
\begin{equation}
\label{eq:sphsymmeestimate1}
\begin{split}
&\int^{u}_{-\infty} |u'|^{p}|{r\partial_u \phi_0}|({u'},v) \,d{u'}+\int_{v_0}^v v'^{-q}|{r\partial_v \phi_0}|(u,{v'})\,d{v'}\\
&\quad\leq C\left[\int_{-\infty}^{u} |u'|^{p}|{r\partial_u \phi_0}|({u'},v_0) \,d{u'}+\int_{v_0}^v v'^{-q}|{r\partial_v \phi_0}|(u_0,{v'})\,d{v'}\right].
\end{split}
\end{equation}

Similarly, for $-1< p < 1$ and $0\leq q< 1-p$, there exists a constant $C=C(M,u_0,v_0,p,q)>0$ such that
\begin{equation}
\label{eq:sphsymmeestimate2}
\begin{split}
&\int_{-\infty}^{u} |u'|^{-p}|{r\partial_u \phi_0}|({u'},v) \,d{u'}+\int_{v_0}^v v'^{q}|{r\partial_v \phi_0}|(u,{v'})\,d{v'}\\
&\quad \leq C\left[\int_{-\infty}^{u} |u'|^{-p}|{r\partial_u \phi_0}|({u'},v_0) \,d{u'}+\int_{v_0}^v v'^{q}|{r\partial_v \phi_0}|(u_0,{v'})\,d{v'}\right].
\end{split}
\end{equation}
\end{proposition}
\begin{proof}
First of all, by (\ref{eq:weth}) and (\ref{eq:wez}), $|\partial_v ({r\partial_u \phi_0})|=\frac{|\partial_u r|}{r}|{r\partial_v \phi_0}|$ and $|\partial_u ({r\partial_v \phi_0})|= \frac{|\partial_v r|}{r}|{r\partial_u \phi_0}|$. In combination with the fundamental theorem of calculus in the $v$ direction, these estimates give
\begin{equation*}
\int_{-\infty}^{u} |u'|^{p}|{r\partial_u \phi_0}|({u'},v) \,d{u'}\leq\int_{-\infty}^{u} |u'|^{p}|{r\partial_u \phi_0}|({u'},v_0) \,d{u'}+\int_{v_0}^v\int_{-\infty}^{u} |u'|^{p}\frac{|\partial_u r|}{r}|{r\partial_v \phi_0}|({u'},{v'})\,d{u'}d{v'}.
\end{equation*}
We can rearrange the terms inside the double integral:
\begin{equation*}
\begin{split}
\int_{v_0}^v\int_{-\infty}^{u} |u'|^{p}\frac{|\partial_u r|}{r}|{r\partial_v \phi_0}|\,d{u'}d{v'} &\leq C\int_{-\infty}^{u} |u'|^{p}\sup_{v_0\leq {v'}\leq v} {v'}^{q}|\partial_u r|({u'},{v'})\left[\int_{v_0}^v {v'}^{-q}|{r\partial_v \phi_0}|({u'},{v'})\,d{v'}\right]\,d{u'}.
\end{split}
\end{equation*}
Since $\partial_ur=-\Omega^2$, we use (\ref{eq:rnOmegainverseexp}) to estimate
\begin{equation*}
|\partial_ur(u,v)|=\Omega^2(u,v)\leq C(v+|u|)^{-2}.
\end{equation*}
Therefore,
\begin{equation*}
\sup_{v_0\leq v'\leq v} {v'}^{q}|\partial_u r|(u',v')\leq C\sup_{v_0\leq v'\leq v} \frac{{v'}^q}{(v'+|u'|)^2}\leq C {|u'|}^{-2+q}.
\end{equation*}
Consequently,
\begin{equation*}
\begin{split}
\int_{v_0}^v\int_{-\infty}^{u} |u'|^{p}\frac{|\partial_u r|}{r}|{r\partial_v \phi_0}|\,d{u'}d{v'}\leq C\int_{-\infty}^{u}|{u'}|^{p+q-2}\left[\int_{v_0}^v {v'}^{-q}|{r\partial_v \phi_0}|({u'},{v'})\,d{v'}\right]\,du'.
\end{split}
\end{equation*}

By interchanging the roles of $u$ and $v$ we can also obtain
\begin{equation*}
\begin{split}
\int_{v_0}^v\int_{-\infty}^{u} v'^{-q}\frac{|\partial_v r|}{r}|{r\partial_u \phi_0}|\,d{u'}d{v'} &\leq C\int_{v_0}^v v'^{-q}\sup_{-\infty\leq u'\leq u} |u'|^{-p}|\partial_v r|(u',v)\left[\int_{-\infty}^{u} |u'|^p|{r\partial_u \phi_0}|({u'},{v'})\,d{u'}\right]\,dv'\\
&\leq C\int_{v_0}^v v'^{-q-2}\left[\int_{-\infty}^{u} |u'|^p|{r\partial_u \phi_0}|({u'},{v'})\,d{u'}\right]\,dv',
\end{split}
\end{equation*}
for $p\geq 0$.

We now apply Lemma \ref{lm:gronwall} with
\begin{align*}
A&=\int_{-\infty}^{u} |u'|^{p}|{r\partial_u \phi_0}|({u'},v_0) \,d{u'}+\int_{v_0}^v v^{-q}|{r\partial_v \phi_0}|(u_0,{v'})\,d{v'},\\
f(u,v)&=\int_{v_0}^v {v'}^{-q}|{r\partial_v \phi_0}|({u},{v'})\,d{v'},\\
g(u,v)&=\int^{u}_{-\infty} |u'|^p|{r\partial_u \phi_0}|({u'},{v})\,du',\\
h(u)&=|u|^{p+q-2},\\
k(v)&=v^{-q-2}.
\end{align*}
Note that $h$ and $k$ are integrable for $q>-1$ and $p+q<1$, so we obtain
\begin{equation*}
\begin{split}
&\int_{-\infty}^{u} |u'|^{p}|{r\partial_u \phi_0}|({u'},v) \,d{u'}+\int_{v_0}^v v^{-q}|{r\partial_v \phi_0}|(u,{v'})\,d{v'}\\
&\quad\leq C\left[\int_{-\infty}^{u} |u'|^{p}|{r\partial_u \phi_0}|({u'},v_0) \,d{u'}+\int_{v_0}^v v^{-q}|{r\partial_v \phi_0}|(u_0,{v'})\,d{v'}\right].
\end{split}
\end{equation*}
We prove (\ref{eq:sphsymmeestimate2}) by interchanging the roles of $u$ and $v$ everywhere above.
\end{proof}

\subsection{$L^{\infty}$ estimates for $\phi_0$ and first-order derivatives}
\label{sec:Linftyfirstorder}

We use the $L^1$ estimates for the derivatives of $\phi_0$ in Proposition \ref{lm:l1} to obtain pointwise estimates for $\phi_0$.
\begin{proposition}
\label{prop:boundsph}
Let $u_0<0$ and $v_0>0$ and fix $p=0$. For $ q <1$, there exists a constant $C=C(M,u_0,v_0,q)>0$, such that for all $(u,v)\in[-\infty,u_0]\times[v_0,\infty)$,
\begin{equation*}
|\phi_0|(u,v)\leq |\phi_0|(-\infty,v)+C F_{0;0,q}[\phi_0],
\end{equation*}
where
\begin{equation*}
F_{0;p,q}[\phi_0]:=\int_{-\infty}^{u_0} |u'|^{p}|\partial_u \phi_0|(u',v_0)\,du'+ \int_{v_0}^{\infty} v'^{-q}|\partial_v\phi_0|(-\infty,v')\,dv'.
\end{equation*}
\end{proposition}
\begin{proof}
Applying the fundamental theorem of calculus in the $u$-direction, together with the first estimate of Proposition \ref{lm:l1} with $p=0$, results in the required estimate:
\begin{equation*}
\begin{split}
|\phi_0|(u,v)&\leq |\phi_0|(-\infty,v)+ \int_{-\infty}^{u} |\partial_u \phi_0|({u'},v)\,d{u'}\\
&\leq |\phi_0|(-\infty,v)+C\left[\int_{-\infty}^{u} |\partial_u \phi_0|({u'},v_0)\,d{u'}+ \int_{v_0}^{v} {v'}^{-q}|\partial_v\phi_0|(-\infty,{v'})\,d{v'}\right],\\
\end{split}
\end{equation*}
where $C>0$ is a uniform constant.
\end{proof}

We can use the transport equations (\ref{eq:weth}) and (\ref{eq:wez}) to moreover establish boundedness of $\partial_{V}\phi_0$ and $\partial_u\phi_0$, where we replace $v$ by $V$ in (\ref{eq:weth}). 
\begin{proposition}
\label{prop:boundsphderv}
Let $u_0<0$ and $v_0>0$. For $ q <1$, there exists a constant $C=C(M,u_0,v_0,q)>0$, such that for all $(u,v)\in[-\infty,u_0]\times[v_0,\infty)$,
\begin{equation}
\label{eq:estdervpsi0}
|\partial_{V}\phi_0|(u,v)\leq C\left[|v^2\partial_v\phi_0|(-\infty,v)+F_{0;0,q}[\phi_0]\right].
\end{equation}
\end{proposition}
\begin{proof}
Recall that
\begin{equation*}
\frac{d V}{dv}=\left(1-\frac{M}{V}\right)^2.
\end{equation*}
We now integrate (\ref{eq:weth}) to find in $(u,V)$ coordinates
\begin{equation}
\label{ineq:eest}
\begin{split}
|\partial_{V} \phi_0|(u,V)&\leq |\partial_{V}\phi_0|(-\infty,V)+C\int_{-\infty}^{u} \left|\frac{\partial_{V}r}{r}\right||{r\partial_u \phi_0}|({u'},V)\,d{u'}\\
&\leq |\partial_{V}\phi_0|(-\infty,V)+C\int_{-\infty}^{u}|{r\partial_u \phi_0}|({u'},V)\,d{u'}\\
&\leq |\partial_{V}\phi_0|(-\infty,V)+  C\left[\int_{-\infty}^{u} |{r\partial_u \phi_0}|({u'},v_0)\,d{u'}+\int_{v_0}^{v} {v'}^{-q}|{r\partial_v \phi_0}|(-\infty,{v'})\,d{v'}\right]\\
\end{split}
\end{equation}
where we used in the second inequality (\ref{eq:part1dervr}) of Proposition \ref{prop:part2dervr} to estimate
\begin{equation*}
|\partial_Vr|\leq C
\end{equation*}
and we applied Proposition \ref{lm:l1} to arrive at the third inequality.

Finally, from (\ref{est:vtoV}) it follows that
\begin{equation*}
|\partial_{V}\phi_0|(-\infty,v)\leq  \left(1-\frac{M}{V}\right)^{-2} |\partial_v\phi_0|(-\infty,v)\leq C v^{2}|\partial_v\phi_0|(u,v). \qedhere
\end{equation*}
\end{proof}

Theorem \ref{thm:estpartialVphi} (ii) follows from Proposition \ref{prop:boundsphderv}.

Similarly, we have uniform boundedness of $\partial_u\phi_0$. Note that unlike $v$, the coordinate $u$ is a regular coordinate at the Cauchy horizon.

\begin{proposition}
\label{prop:boundsphderu}
Let $u_0<0$ and $v_0>0$. For $ q <1$, there exists a constant $C=C(M,u_0,v_0,q)>0$, such that for all $(u,v)\in[-\infty,u_0]\times[v_0,\infty)$,
\begin{equation}
\label{eq:estdupsi0}
|\partial_{u}\phi_0|(u,v)\leq C\left[|\partial_u\phi_0|(u,v_0)+(v_0+|u|)^{-2}F_{0;0,q}[\phi_0]\right].
\end{equation}
\end{proposition}
\begin{proof}
By interchange the roles of $u$ and $V$ in the proof of Proposition \ref{prop:boundsphderv}, together with
\begin{equation*}
\sup_{v_0<v'\leq v} |\partial_ur|(u,v')=\Omega^2(u,v_0)\leq C(v_0+|u|)^{-2}, 
\end{equation*}
we immediately arrive at (\ref{eq:estdupsi0}).
\end{proof}

\subsection{$L^{\infty}$ estimates for $\partial_V^2\phi_0$}
\label{sec:Linftysecondorder}
From (\ref{eq:rntortoise}) it follows that
\begin{equation}
\label{eq:relationutoU}
U=M-r(u,v_0)\leq C(v_0+|u|)^{-1}\leq C|u|^{-1}.
\end{equation}
Furthermore, since the initial data in Proposition \ref{prop:wellposedness} for $\phi_0$ are smooth, we can apply Taylor's formula at $U=0$, together with (\ref{eq:relationutoU}) to obtain
\begin{equation*}
\label{eq:sphsymdatain}
\partial_U(r\phi_0)(u,v_0)=-MH_0+\mathcal{O}(|u|^{-1})
\end{equation*}
where $H_0$ is the zeroth Aretakis constant, a conserved quantity for $\phi_0$ along $\mathcal{H}^+$; see the discussion in Section \ref{sec:priceslaw}.

Note that (\ref{eq:sphsymdatain}) is equivalent to
\begin{equation}
\label{eq:mainestduphi0}
\partial_u(r\phi_0)(u,v_0)=-M^3H_0|u|^{-2}+\log |u|\mathcal{O}(|u|^{-3}),
\end{equation}
by the estimates in Proposition \ref{prop:metricc1rn}.

Let us consider the following pointwise decay assumption for $\phi_0$ along $\mathcal{H}^+$.
\begin{align}
\label{eq:sphsymdataout}
|r \phi_0|(-\infty,v)&\leq C_1v^{-s},
\end{align}
where $s>0$ and $0<C_1<\infty$. 

\begin{proposition}
\label{prop:asymptoticspartialuphi0}
Let the initial data for $\phi_0$ satisfy (\ref{eq:sphsymdataout}). Then we can expand in the region $|u|\geq |u_1|$, for $|u_1|$ suitably large,
\begin{equation}
\label{eq:lowerboundderupsi0}
\partial_u\phi_0(u,v)=-M^2H_0|u|^{-2}+ \mathcal{O}(|u|^{-2-\min\{s,1\}})+v^{-s}\mathcal{O}(|u|^{-2}).
\end{equation}
\end{proposition}
\begin{proof}
We rewrite (\ref{eq:wez}) in the form
\begin{equation}
\label{eq:waveqalt}
\partial_v\partial_u(r\phi_0)=(\partial_u\partial_vr) \phi_0.
\end{equation}
By the fundamental theorem of calculus, we have that
\begin{equation}
\label{eq:partialurpsi0}
\partial_u(r\phi_0)(u,v)=\partial_u(r\phi_0)(u,v_0)+\int_{v_0}^v\frac{\partial_u\partial_vr}{r}r\phi_0(u,v')\,dv'.
\end{equation}
Consequently, by applying the fundamental theorem of calculus once more and using (\ref{eq:mainestduphi0}), we obtain
\begin{equation}
\label{est:mainestlowerboundpsi0}
\begin{split}
r\phi_0(u,v)&=r\phi_0(-\infty,v)+\int_{-\infty}^u\partial_u(r\phi_0)(u',v)\,du'\\
&=r\phi_0(-\infty,v)+\int_{-\infty}^u\partial_u(r\phi_0)(u',v_0)\,du'+\int_{-\infty}^u\int_{v_0}^v\frac{\partial_u\partial_vr}{r}r\phi_0(u',v')\,dv'du'\\
&=r\phi_0(-\infty,v)-M^3H_0|u|^{-1}+\int_{-\infty}^u\int_{v_0}^v\frac{\partial_u\partial_vr}{r}r\phi_0(u',v')\,dv'du'+\log |u|\mathcal{O}(|u|^{-2}).
\end{split}
\end{equation}
We have that $0<\partial_u\partial_v r<C(v+|u|)^{-3}$. Let us make the bootstrap assumption
\begin{equation}
\label{ba:psi0lowerbound}
|r\phi_0(u,v)-r\phi_0(-\infty,v)|\leq 2M^3|H_0||u|^{-1},
\end{equation}
for all $v\in [v_0,\infty)$. If we use (\ref{ba:psi0lowerbound}) together with (\ref{est:mainestlowerboundpsi0}) and moreover apply (\ref{eq:sphsymdataout}), we obtain
\begin{equation*}
\begin{split}
|r\phi_0(u,v)-r\phi_0(-\infty,v)|&\leq M^3|H_0||u|^{-1}+\mathcal{O}(|u|^{-2}\log |u|)+C\int_{-\infty}^u\int_{v_0}^v\frac{1}{|u'|(v'+|u'|)^3}\,dv'du'\\
&\quad+C\int_{-\infty}^u\int_{v_0}^v\frac{1}{v'^s(v'+|u')^3}\,dv'du'\\
&\leq M^3|H_0||u|^{-1}+\mathcal{O}(|u|^{-1-\min\{1,s\}}).
\end{split}
\end{equation*}
The above estimate improves the bootstrap assumption (\ref{ba:psi0lowerbound}), if we restrict to $|u|>|u_1|$, with $|u_1|$ suitably large.

Applying (\ref{ba:psi0lowerbound}) and (\ref{eq:sphsymdataout}), we can conclude by (\ref{eq:partialurpsi0}), using again that $0<\partial_u\partial_v r<C(v+|u|)^{-3}$, 
\begin{equation}
\label{eq:mainlowerbound}
\begin{split}
\left|\partial_u(r\phi_0)(u,v)+M^3H_0|u|^{-2}\right|&\leq C\int_{v_0}^v \frac{1}{{v'}^s(v'+|u'|)^3}\,dv'+C\int_{v_0}^v \frac{1}{|u'|(v'+|u'|)^3}\,dv'\\
&\leq C|u|^{-2-\min\{1,s\}},
\end{split}
\end{equation}
for $|u|>|u_1|$, where $|u_1|$ is chosen suitably large.

The lower bound (\ref{eq:mainlowerbound}), together with (\ref{ba:psi0lowerbound}) and (\ref{eq:sphsymdataout}), therefore results in the estimate
\begin{equation*}
\begin{split}
\partial_u\phi_0&=r^{-1}\partial_u(r\phi_0)-r^{-1}(\partial_ur) \phi_0\\
&=-M^2H_0|u|^{-2}+\mathcal{O}(|u|^{-2-\min\{s\}})+v^{-s}\mathcal{O}(|u|^{-2}). \qedhere
\end{split}
\end{equation*}
\end{proof}

The next step is to use the refined estimate for $\partial_u\phi$ to obtain an estimate for $\partial_V^2\phi(u,v)$. We will need an additional assumption on the initial data for $\phi_0$ along $\mathcal{H}^+$:
\begin{equation}
\label{eq:sphsymdataout2}
|\partial_v\phi_0|(-\infty,v)\leq C_2v^{-2},
\end{equation}
where $C_2>0$ is a constant.

\begin{proposition}
\label{prop:blowup2nddersphsymm}
Let the initial data for $\phi_0$ satisfy (\ref{eq:sphsymdataout}) and (\ref{eq:sphsymdataout2}). For $q<1$ and $|u_1|$ suitably large, there exists a uniform constant $C=C(M,u_0,v_0,u_1,q)>0$ such that
\begin{equation}
\label{est:lowerboundder2vpsi0}
\left|r\partial_V^2\phi_0(u,V)-M\partial_V^2\phi_0(-\infty,V)+2H_0\log v\right|\leq CF_{0;0,q}[\phi_0]+C_0,
\end{equation}
for all $u\in (-\infty,u_1]$, where $C_0>0$ depends on the initial data for $\phi_0$.
\end{proposition}
\begin{proof}
Consider the transport equation (\ref{eq:weth}), where we replace $v$ by $V$. We apply $\partial_V$ to both sides of the equation to arrive at
\begin{equation}
\label{eq:partialudv2phi}
\begin{split}
\partial_u(r\partial_V^2\phi_0)&=-\partial_V^2r\partial_u\phi_0-\partial_Vr\partial_V\partial_u\phi_0-\partial_u(\partial_Vr \partial_V\phi_0)\\
&=-\partial_V^2r \partial_u\phi_0-\partial_u\partial_Vr \partial_V\phi_0-2r^{-1}\partial_Vr \partial_u(r\partial_V\phi_0)+2r^{-1}\partial_Vr\partial_ur \partial_V\phi_0\\
&=-\partial_V^2r \partial_u\phi_0-\partial_u\partial_Vr \partial_V\phi_0+2r^{-1}(\partial_Vr)^2 \partial_u\phi_0+2r^{-1}\partial_Vr\partial_ur \partial_V\phi_0,
\end{split}
\end{equation}
where we used the transport equation (\ref{eq:weth}) in the last equality. 

We integrate in the $u$-direction to obtain
\begin{equation*}
\begin{split}
r\partial_V^2\phi_0(u,V)&=M\partial_V^2\phi_0(-\infty,V)-\int_{-\infty}^{u}(\partial_V^2r)\partial_u\phi_0(u',V)\,du'\\
&\quad+\int_{-\infty}^u-\partial_u\partial_Vr \partial_V\phi_0+2r^{-1}(\partial_Vr)^2 \partial_u\phi_0+2r^{-1}\partial_Vr\partial_ur \partial_V\phi_0\,du'.
\end{split}
\end{equation*}
By applying Proposition \ref{lm:l1} and Proposition \ref{prop:boundsphderv}, we can estimate
\begin{equation*}
\begin{split}
\int_{-\infty}^u-\partial_u\partial_Vr \partial_V\phi_0+2r^{-1}(\partial_Vr)^2 \partial_u\phi_0+2r^{-1}\partial_Vr\partial_ur \partial_V\phi_0\,du'&\leq \sup_{-\infty\leq u'\leq u}|\partial_V\phi_0|(u',V)+CF_{0;0,q}[\phi_0]\\
&\leq C|\partial_V\phi_0|(-\infty,V)+CF_{0;0,q}[\phi_0].
\end{split}
\end{equation*}
We rearrange the terms above to conclude that
\begin{equation*}
\begin{split}
\left|r\partial_V^2\phi_0(u,V)-M\partial_V^2\phi_0(-\infty,V)+\int_{-\infty}^{u}(\partial_V^2r)\partial_u\phi_0(u',V)\,du'\right|&\leq  C|\partial_V\phi_0|(-\infty,V)+CF_{0;0,q}[\phi_0],
\end{split}
\end{equation*}
where $C>0$ is a uniform constant.

By Proposition \ref{prop:part2dervr} we can expand
\begin{equation*}
\Omega^{2}(u_0,v)\partial_V^2 r(u,v) = 2\frac{|u|-|u_0|}{(v+|u|)^2}+(|u|-|u_0|)\log(v+|u|)\mathcal{O}_1((v+|u|)^{-3}),
\end{equation*}
where $C>0$ is a uniform constant. Hence, by moreover using Proposition \ref{prop:asymptoticspartialuphi0} and (\ref{eq:rnOmegainverseexp}), we obtain
\begin{equation*}
\begin{split}
\int_{-\infty}^{u}(\partial_V^2r)\partial_u\phi_0(u',V)\,du'&=-2H_0\Omega^{-2}(u_0,v) \int_{-\infty}^u \frac{1}{M^2|u'|(v+|u'|)^2}\,du'+\mathcal{O}_1(1)\\
&=-2H_0\log v +\mathcal{O}_1(1).
\end{split}
\end{equation*}

Applying moreover (\ref{eq:estdervpsi0}) and (\ref{eq:sphsymdataout2}), we arrive at the final estimate
\begin{equation*}
\begin{split}
\left|r\partial_V^2\phi_0(u,V)-M\partial_V^2\phi_0(-\infty,V)-2H_0\log v\right|&\leq CF_{0;0,q}[\phi_0]+C
\end{split}
\end{equation*}
for $|u|>|u_1|$, with $|u_1|$ suitably large.
\end{proof}

To determine whether $\partial_V^2\phi_0(u,v)$ is bounded, we need to find out whether the logarithmic term in (\ref{est:lowerboundder2vpsi0}) gets cancelled by the leading order term in $\partial_V^2\phi_0(-\infty,v)$. We need a more precise assumption on the asymptotics of $\phi_0$, compared to assumption (\ref{eq:sphsymdataout2}).

We consider the following asymptotic behaviour for $\phi_0$ along the event horizon, motivated by the numerics in \cite{Lucietti2013} (see the discussion in Section \ref{sec:priceslaw}):
\begin{equation}
\label{eq:asympoticsphiinitial}
\phi_0(-\infty,v)=-\frac{M^2H_0}{v}+\frac{M^3 H_0}{v^2}\log v+\mathcal{O}_{2}(v^{-2}),
\end{equation}
Note that this expansion implies in particular the assumption (\ref{eq:sphsymdataout}) with $q=1$.

Consequently, by using the expression (\ref{eq:rnOmegainverseexp}) for $(\partial_vr)^{-1}(u_0,v)=\Omega^{-2}(u_0,v)$, we obtain
\begin{equation}
\label{eq:asymppartialVphi}
\begin{split}
\partial_V\phi_0(-\infty,v)&=(\partial_vr)^{-1}(u_0,v) \left(\frac{M^2H_0}{v^2}-\frac{2M^3 H_0}{v^3}\log v+ \mathcal{O}_{1}(v^{-3})\right)\\
&=M^{-2}(v+|u_0|)^2\left(\frac{M^2H_0}{ v^2}-\frac{2 M^3H_0}{v^3}\log v+ \mathcal{O}_{1}(v^{-3})\right)\\
&\quad+4M(v+|u_0|)\log(v+|u_0)\frac{H_0}{v^2}+\mathcal{O}_{1}(v^{-1})\\
&=H_0\left(\frac{v+|u_0|}{v}\right)^2+\frac{2H_0M}{v}\log v+\mathcal{O}_{1}(v^{-1})\\
&=H_0+\frac{2H_0M}{v}\log v+\mathcal{O}_{1}(v^{-1}).
\end{split}
\end{equation}
The above expansion implies the assumption (\ref{eq:sphsymdataout2}).

We take one more $v$-derivative to obtain
\begin{equation*}
\partial_v\partial_V\phi_0(-\infty,v)=-\frac{2H_0M}{v^2}\log v+\mathcal{O}(v^{-2}).
\end{equation*}
Hence,
\begin{equation}
\label{eq:asymppartialV2phi}
M\partial_V^2\phi_0(-\infty,v)=-2H_0 \log v+\mathcal{O}(1).
\end{equation}

We can use the asymptotics in (\ref{eq:asymppartialVphi}) together with the asymptotics in Proposition \ref{prop:asymptoticspartialuphi0} to first of all obtain the asymptotics of $\partial_V\phi_0(u,v)$.
\begin{proposition}
\label{cor:asymppartialVphieverywhere}
Let the initial data for $\phi_0$ satisfy the asymptotic behaviour (\ref{eq:asympoticsphiinitial}). Then for $|u_1|$ suitably large,
\begin{align*}
\partial_V\phi_0(u,v)&=H_0+\frac{2H_0}{v}\log v+\mathcal{O}(|u|^{-1})+\mathcal{O}(v^{-1}),\\
\partial_V(r\phi_0)(u,v)&=MH_0+\frac{2MH_0}{v}\log v+ \mathcal{O}(v^{-1}),
\end{align*}
for all $|u|>|u_1|$.
\end{proposition}
\begin{proof}
By applying (\ref{eq:part1dervr}) of Proposition \ref{prop:part2dervr} we can estimate
\begin{equation*}
\begin{split}
r\partial_V\phi_0(u,v)&=M\partial_V\phi_0(-\infty,v)-\int_{-\infty}^u \partial_Vr \partial_u\phi_0(u',v)\,du'\\
&=M\partial_V\phi_0(-\infty,v)-(\phi(u,v)-\phi(-\infty,v))+\mathcal{O}_2((v+|u|)^{-1})\\
&=MH_0+\frac{2H_0M}{v}\log v+\mathcal{O}(|u|^{-1})+\mathcal{O}(v^{-1}),
\end{split}
\end{equation*}
where we used (\ref{eq:asymppartialVphi}) and Proposition \ref{prop:asymptoticspartialuphi0} to arrive at the last equality. We can rewrite the above expressions to obtain
\begin{equation*}
\begin{split}
\partial_V(r\phi_0)(u,v)&=\partial_Vr \phi_0(u,v)+r\partial_V\phi_0(u,v)\\
&=M\partial_V\phi_0(-\infty,v)+\phi(-\infty,v)+\mathcal{O}_2((v+|u|)^{-1})\\
&=MH_0+\frac{2H_0M}{v}\log v+\mathcal{O}(v^{-1}).\qedhere
\end{split}
\end{equation*}
\end{proof}
\begin{remark}
Recall from the discussion in Section \ref{sec:priceslaw} that the limit $\lim_{r\to M} \partial_r(r\phi_0)(v,r)$ is independent of $v$ and is therefore conserved along $\mathcal{H}^+$. Indeed, we have that
\begin{equation*}
MH_0=\lim_{r\to M} \partial_r(r\phi_0)(v,r).
\end{equation*}
By Proposition \ref{cor:asymppartialVphieverywhere}, $\phi_0$ is $C^1$-extendible beyond $\mathcal{CH}^+$, so in outgoing Eddington-Finkelstein coordinates $(u,r,\theta,\varphi)$ on $\mathcal{M}_{\rm int}$ it follows similarly that the quantity $\lim_{r\to M} \partial_r(r\phi_0)(u,r)$ is independent of $u$ and therefore defines another constant:
\begin{equation*}
MH_0'=\lim_{r\to M} \partial_r(r\phi_0)(u,r).
\end{equation*}
Remarkably, by using the asymptotics in Proposition \ref{cor:asymppartialVphieverywhere}, it turns out that $H_0=H_0'$.
\end{remark}

Since the leading order term of $M\partial_V^2\phi(-\infty,v)$ in (\ref{eq:asymppartialV2phi}) \emph{exactly} cancels out the logarithmic term in (\ref{est:lowerboundder2vpsi0}), we can conclude that $\partial_V^2\phi$ is bounded at the Cauchy horizon.
\begin{proposition}
\label{cor:c2boundphi0}
Let the initial data for $\phi_0$ satisfy the asymptotic behaviour (\ref{eq:asympoticsphiinitial}). Then
\begin{equation*}
|\partial_V^2\phi_0|(u,v)\leq C_0,
\end{equation*}
everywhere in $\mathcal{M}\cap D_{u_0,v_0}$, where $C_0>0$ is a constant depending on the initial data.
\end{proposition}
\begin{proof}
The estimate (\ref{est:lowerboundder2vpsi0}), together with (\ref{eq:asymppartialV2phi}) imply immediately that
\begin{equation*}
|\partial_V^2\phi_0|(u,v)\leq C_0,
\end{equation*}
for $|u|\geq |u_1|$, where $|u_1|$ is suitably large, where $C_0>0$ is a constant depending on the initial data. 

We now integrate the expression for $\partial_u(r\partial_V^2\phi_0)$ in (\ref{eq:partialudv2phi}) from $u_1$ to $u_0$ to obtain
\begin{equation*}
\begin{split}
\left|r\partial_V^2\phi_0(u,V)-M\partial_V^2\phi_0(u_1,V)+\int_{u_1}^{u}(\partial_V^2r)\partial_u\phi_0(u',V)\,du'\right|&\leq  C|\partial_V\phi_0|(-\infty,V)+CF_{0;0,q}[\phi_0],
\end{split}
\end{equation*}
where we repeated the arguments in Proposition \ref{est:lowerboundder2vpsi0}, with $u_1$ replacing $-\infty$. As the interval $[u_1,u_0]$ has a finite length, we can use (\ref{eq:rnOmegainverseexp}) together with Proposition \ref{lm:l1} to further estimate
\begin{equation*}
\begin{split}
\int_{u_1}^{u}(\partial_V^2r)\partial_u\phi_0(u',V)\,du'&\leq C \int_{u_1}^u |u| |\partial_u\phi_0|(u',V)\,du'\\
&\leq |u_1|\int_{u_1}^u |\partial_u\phi_0|(u',V)\,du',\\
&\leq C|u_1| F_{0;0,q}[\phi_0].
\end{split}
\end{equation*}

We can now conclude that for all $u\in (-\infty,u_0)$
\begin{equation*}
\begin{split}
|\partial_V^2\phi_0|(u,V)&\leq |\partial_V^2\phi_0|(u_1,V)+C|\partial_V\phi_0|(-\infty,V)+C|u_1|F_{0;0,q}[\phi_0]\\
&\leq C_0. \qedhere
\end{split}
\end{equation*}
\end{proof}

\begin{remark}
Let $H_0\neq 0$. If the leading-order term in $M\partial_V^2\phi_0(-\infty,v)$ is \uline{any} different from (\ref{eq:asymppartialV2phi}), the estimate (\ref{est:lowerboundder2vpsi0}) implies \emph{blow-up} of $|\partial_V^2\phi_0|(u,v)$ at $\mathcal{CH}^+$, for $|u|>|u_1|$. 

If $H_0=0$, the precise constant in the leading-order term of the asymptotics of $M\partial_V^2\phi_0(-\infty,v)$ does not matter and we only need to assume uniform boundedness of $M\partial_V^2\phi_0(-\infty,v)$ to conclude that $\partial_V^2\phi_0(u,v)$ is uniformly bounded for all $u\in (-\infty,u_0)$.
\end{remark}

We have now proved Theorem \ref{thm:bound2nddersphsymm}.

\subsection{Regularity at $\mathcal{CH}^+$}
\label{sec:regularityatCH+}
From the above sections, we can infer that derivatives of $\phi_0$ up to second order remain bounded at $\mathcal{CH}^+$, with respect to $(u,V,\theta,\varphi)$ coordinates that are regular at $\mathcal{CH}^+$, if we assume appropriate asymptotics along $\mathcal{H}^+$ and suitable regularity along $\uline{H}_{v_0}$. Under these assumptions, we will furthermore show in this section that $\phi_0$ can be extended as a $C^2$ function beyond $\mathcal{CH}^+$.

We first show that $\phi_0$ can be extended as a $C^0$ function beyond $\mathcal{CH}^+$.
\begin{proposition}
\label{prop:sphsymmc0}
Let the initial data for $\phi_0$ satisfy $F_{0;p,0}[\phi_0]<\infty$, for $0\leq p<1$. Then $\phi_0$ can be extended as a $C^0$ function beyond $\mathcal{CH}^+$.
\end{proposition}
\begin{proof}
We first need to show that we can extend $\phi_0$ as a function to the entire region $D_{u_0,v_0}$. That is to say, $\phi_0$ is well-defined on $\mathcal{CH}^+$.

Let $V_0:=V(v_0)$ and define,
\begin{equation*}
\phi_0(u,M):=\phi_0(u,V_0)+\lim_{V\to M} \int_{V_0}^V\partial_V\phi_0(u,V')\,dV'.
\end{equation*}
Since,
\begin{equation*}
\begin{split}
\int_{V_0}^M\partial_V\phi_0(u,V')\,dV'&=\int_{v_0}^{\infty}\partial_v\phi_0(u,v')\,dv'\\
&\leq CF_{0;p,0}[\phi_0],
\end{split}
\end{equation*}
for some $0\leq p<1$, by Proposition \ref{eq:sphsymmeestimate1}, we can conclude that $\phi_0(u,M)$ is well-defined, for all $-\infty<u<u_0$, if $F_{0;p,0}[\phi_0]<\infty$.

We will now show that $\phi_0$ is a continuous function on $D_{u_0,v_0}$. Let $(u_1,V_1)$ satisfy $-\infty\leq u_1 <u_0$ and $V_0\leq V_1\leq M$. Without loss of generality, we assume $u>u_1$ and $V<V_1$, so that
\begin{equation*}
\begin{split}
|\phi_0(u_1,V_1)-\phi_0(u,V)|&\leq |\phi_0(u_1,V)-\phi_0(u,V)|+|\phi_0(u_1,V)-\phi_0(u_1,V_1)|\\
&\leq \int_{u_1}^u |\partial_u\phi_0|(u',V)\,du'+\int_{V}^{V_1}|\partial_V\phi_0|(u_1,V')\,dV'.
\end{split}
\end{equation*}

Since,
\begin{align*}
\int_{u_1}^u |\partial_u\phi_0|(u',V)\,du'&\leq \tilde{C}(u,u_1)F_{0;p,0}[\phi_0],\\
\int_{V}^{V_1}|\partial_V\phi_0|(u_1,V')\,dV'&\leq CF_{0;p,0}[\phi_0],
\end{align*}
where $C,\tilde{C}>0$ are constants and $\tilde{C}>0$ depends on $u$ and $u_1$, we can conclude that
\begin{equation*}
\begin{split}
\lim_{(u,V)\to(u_1,V_1)}|\phi_0(u_1,V_1)-\phi_0(u,V)|&\leq \lim_{(u,V)\to (u_1,V_1)}\int_{u_1}^u |\partial_u\phi_0|(u',V)\,du'\\
&\quad+\lim_{V\to V_1}\int_{V}^{V_1}|\partial_V\phi_0|(u_1,V')\,dV'\\
&=0.
\end{split}
\end{equation*}
The function $\phi_0$, extended to $D_{u_0,v_0}$, is therefore continuous.
\end{proof}

Under stronger assumptions on the initial data, we can in fact conclude that $\phi_0$ is $C^1$-extendible. We will show that $\partial_V\phi_0$ is $C^1$-extendible. Continuity of $\partial_u\phi_0$ follows from commuting $\square_g$ with $\partial_u$. See also Section \ref{sec:unifboundcomm}.
\begin{proposition}
\label{prop:sphsymmc1}
Let the initial data for $\phi_0$ satisfy $F_{0;p,0}[\phi_0]<\infty$ for some $p>0$, such that moreover $\lim_{v\to \infty}v^2\partial_v\phi_0|_{\mathcal{H}^+}(v)$ is well-defined. Then $\partial_V\phi_0$ can be extended as a $C^0$ function beyond $\mathcal{CH}^+$.
\end{proposition}
\begin{proof}
We repeat the arguments of Proposition \ref{prop:sphsymmc0} with $\phi_0$ replaced by $\partial_V\phi_0$. We can therefore conclude that $\partial_V \phi_0$ is continuous everywhere in $D_{u_0,v_0}$, including at $\mathcal{CH}^+$, if we can bound
\begin{align}
\label{eq:c11}
\left|\int_{u_1}^u \partial_u\partial_V\phi_0(u',v)\,du'\right|&\leq \tilde{C}(u,u_1)D_0,\\
\label{eq:c12}
\left|\int_{v_0}^{\infty}\partial_v\partial_V\phi_0(u_1,v')\,dv'\right|&\leq CD_0,
\end{align}
with $u_0\leq u_1<u$ and $v_0\leq v<\infty$, where $D_0>0$ is a constant depending on the initial data. 

In order to estimate (\ref{eq:c11}), we use that $\partial_u(r\partial_V\phi_0)=-\partial_Vr \partial_u\phi_0$ and split
\begin{equation*}
\begin{split}
\left|\int_{u_1}^u \partial_u\partial_V\phi_0(u',v)\,du'\right|&\leq \left|\int_{u_1}^u \partial_u(r\partial_V\phi_0)(u',v)\,du'\right|+C\int_{u_1}^u |\partial_u\phi_0|\,du'\\
&\leq C\int_{u_1}^u |\partial_u\phi_0|\,du'\leq \tilde{C}(u,u_1)F_{0;p,0}[\phi_0],
\end{split}
\end{equation*}
for $0\leq p<1$. We can estimate (\ref{eq:c12}) by first computing $\partial_u(\partial_v(r\partial_V\phi_0))$. We have that
\begin{equation*}
\begin{split}
\partial_v(\partial_u(r\partial_V\phi_0))&=-\partial_v(\partial_V r \partial_u\phi_0)\\
&=-\partial_v(r^{-1}\partial_Vr)r\partial_u\phi_0-r^{-1}\partial_Vr \partial_v(r\partial_u\phi_0)\\
&=-\partial_v(r^{-1}\partial_Vr)r\partial_u\phi_0+r^{-1}\partial_Vr \partial_ur\partial_v\phi_0.
\end{split}
\end{equation*}
Therefore, by using the expression for $\partial_v\partial_Vr$ from Proposition \ref{prop:part2dervr}, we can estimate
\begin{equation*}
\begin{split}
\left|\int_{v_0}^{\infty}\partial_v\partial_V\phi_0(u_1,v')\,dv'\right|&=\left|\int_{v_0}^{\infty}\partial_v\partial_V\phi_0(-\infty,v')\,dv'+\int_{-\infty}^u\int_{v_0}^{\infty}\partial_u\partial_v\partial_V\phi_0(u',v')\,dv'\,dv'du'\right|\\
&\leq \left|\int_{v_0}^{\infty}\partial_v\partial_V\phi_0(-\infty,v')\,dV'\right|\\
&\quad+C\int_{-\infty}^{u_1}\int_{v_0}^{\infty}|u|(v+|u|)^{-2}|\partial_u\phi_0|+(v+|u|)^{-2}\partial_v\phi_0\,dv'du'\\
&\leq \left|\lim_{v\to \infty}\partial_V\phi_0(-\infty,v)-\partial_V\phi_0(-\infty,v_0)\right|+CF_{0;p,0}[\phi_0],
\end{split}
\end{equation*}
for $0<p<1$.

If we assume that $F_{0;p,0}[\phi_0]<\infty$ for $p>0$ and that $\lim_{v\to \infty}\partial_V\phi_0(-\infty,v)$ is well-defined, then we can conclude that $\partial_V\phi_0$ is continuous at $\mathcal{CH}^+$.
\end{proof}

In order to conclude that $\phi_0$ is $C^2$-extendible beyond $\mathcal{CH}^+$, we are left with showing that $\partial_V^2\phi_0$, $\partial_V\partial_u\phi_0$ and $\partial_u^2\phi_0$ can be extended as continuous functions beyond $\mathcal{CH}^+$. Continuous extendibility of $\partial_u^2\phi_0$ follows from the estimates in Section \ref{sec:unifboundcomm}. Continuous extendibility of $\partial_V\partial_u\phi_0$ follows by recalling that
\begin{equation*}
\begin{split}
r\partial_u\partial_V\phi_0&=\partial_u(r\partial_V\phi_0)-\partial_u r \partial_V\phi_0\\
&=-\partial_Vr \partial_u \phi_0-\partial_u r \partial_V\phi_0,
\end{split}
\end{equation*}
where we used (\ref{eq:weth}). We therefore only have to show that $\partial_V^2\phi_0$ can be continuously extended beyond $\mathcal{CH}^+$. We will need to assume the asymptotics in (\ref{eq:asympoticsphiinitial}).

\begin{proposition}
\label{prop:sphsymmc2}
Let the initial data of $\phi_0$ satisfy the asymptotic behaviour (\ref{eq:asympoticsphiinitial}) and assume moreover that
\begin{equation}
\label{eq:limitwelldefined}
\lim_{v\to \infty}M\partial_V^2\phi_0(-\infty,v)+2H_0\log v\quad \textnormal{is well-defined.}
\end{equation}
Then $\partial_V^2\phi_0$ can be extended as a $C^0$ function beyond $\mathcal{CH}^+$.
\end{proposition}
\begin{proof}
As in the proof of Proposition \ref{prop:sphsymmc1}, continuity of $\partial_V^2\phi_0$ follows from the inequalities below:
\begin{align}
\label{eq:c21}
\left|\int_{u_2}^u \partial_u\partial_V^2\phi_0(u',v)\,du'\right|&\leq \tilde{C}(u,u_2)D_0,\\
\label{eq:c22}
\left|\int_{v_0}^{\infty}\partial_v\partial_V^2\phi_0(u_1,v')\,dv'\right|&\leq CD_0,
\end{align}
By integrating (\ref{est:lowerboundder2vpsi0}) in $u$, we conclude that
\begin{equation*}
\left|\int_{u_2}^u \partial_u(r\partial_V^2\phi_0(u',v))\,du'\right|\leq C\sup_{u_2\leq u'\leq u}|\partial_V\phi_0|(u',V)+CF_{0;0,q}[\phi_0]\leq CD_0.
\end{equation*}
The estimate (\ref{eq:c21}) immediately follows. We are left with proving (\ref{eq:c22}). 

Let us first restrict to the region $u<u_1$, where $u_1$ is chosen suitably large, in accordance with Proposition \ref{prop:asymptoticspartialuphi0}. We can estimate
\begin{equation*}
\left|\int_{v_0}^{\infty}\partial_v(r\partial_V^2\phi_0)(u_1,v')\,dv'\right|=\left|\int_{v_0}^{\infty}\left[\partial_v(r\partial_V^2\phi_0)(-\infty,v')+\int_{-\infty}^{u_1}\partial_u\partial_v(r\partial_V^2\phi_0)(u',v')\,du'\right]\,dv'\right|.
\end{equation*}

We differentiate (\ref{eq:partialudv2phi}) in $v$ to obtain
\begin{equation*}
\begin{split}
\partial_v[\partial_u(r\partial_V^2\phi_0)]&=-\partial_v(r^{-1}\partial_V^2 r)r\partial_u\phi_0-r^{-1}\partial_V^2r \partial_v(r\partial_u\phi_0)-\partial_v\partial_u\partial_Vr \partial_V\phi_0-\partial_u\partial_Vr \partial_v\partial_V\phi_0\\
&\quad+2\partial_v(r^{-2}(\partial_Vr)^2)r\partial_u\phi+2r^{-2}(\partial_Vr)^2\partial_v(r\partial_u\phi_0)+2\partial_v(r^{-1}\partial_Vr\partial_u r)\partial_V\phi_0\\
&\quad+2r^{-1}\partial_V r\partial_u r \partial_v\partial_V\phi_0\\
&=\left[-\partial_v(r^{-1}\partial_V^2 r)+2\partial_v(r^{-2}(\partial_Vr)^2)\right]r\partial_u\phi_0\\
&\quad+\left[-\partial_v\partial_u\partial_Vr+2\partial_v(r^{-1}\partial_Vr\partial_u r)\right] \partial_V\phi_0\\
&\quad\left[2r^{-1}\partial_V r\partial_u r-\partial_u\partial_Vr\right] \partial_v\partial_V\phi_0+\left[2r^{-2}(\partial_Vr)^2\partial_vr-r^{-1}\partial_V^2r\partial_vr\right]\partial_v\phi_0
\end{split}
\end{equation*}
We can bound the integral over $u$ and $v$ of most of the terms on the right-hand side to obtain
\begin{equation*}
\begin{split}
&\left|\int_{v_0}^{\infty}\left[\partial_v(r\partial_V^2\phi_0)(-\infty,v')+\int_{-\infty}^{u_1}\partial_u\partial_v(r\partial_V^2\phi_0)(u',v')\,du'\right]\,dv'\right|\\
&\leq \left|\int_{v_0}^{\infty}\left[\partial_v(r\partial_V^2\phi_0)(-\infty,v')-\int_{-\infty}^{u_1}\partial_v\partial_V^2 r\partial_u\phi_0+\partial_u\partial_v\partial_Vr\partial_V\phi_0\,du'\right]\,dv'\right|+C_0,
\end{split}
\end{equation*}
where $C_0>0$ is a constant that depends on the initial data. We used in particular the estimates in Proposition \ref{prop:blowup2nddersphsymm} to deal with the $\partial_v\partial_V\phi_0$ terms.

We differentiate the terms in Proposition \ref{prop:part2dervr} to obtain
\begin{align*}
\partial_V^2r&=2M^{-2}(|u|-|u_0|)\left(\frac{v+|u_0|}{v+|u|}\right)^2+|u|\log(v+|u|)\mathcal{O}_1((v+|u|)^{-1}),\\
\partial_v\partial_V^2r&=4M^{-2}\frac{v+|u_0|}{(v+|u|)^3}(|u|-|u_0|)^2+|u|\log(v+|u|)\mathcal{O}((v+|u|)^{-2}),\\
\partial_u\partial_v\partial_Vr&=-2(v+|u|)^{-2}+4|u|(v+|u|)^{-3}+\log(v+|u|)\mathcal{O}((v+|u|)^{-2}).
\end{align*}
By using the expressions for $\partial_u\phi_0$ and $\partial_V\phi_0$ from Proposition \ref{prop:asymptoticspartialuphi0} and Proposition \ref{cor:asymppartialVphieverywhere}, we can write
\begin{align*}
\partial_v\partial_V^2r\partial_u\phi_0&=-4H_0v(v+|u|)^{-3}+\ldots,\\
\partial_u\partial_v\partial_Vr\partial_V\phi_0&=-2H_0(v+|u|)^{-2}+4H_0|u|(v+|u|)^{-3}+\ldots\\
&=2H_0(v+|u|)^{-2}-4H_0v(v+|u|)^{-3}+\ldots,
\end{align*}
where $\ldots$ indicates all terms that are bounded after integrating in $u$ and subsequently in $v$.
Hence,
\begin{equation*}
\partial_v\partial_V^2r\partial_u\phi_0+\partial_u\partial_v\partial_Vr\partial_V\phi_0=2H_0(v+|u|)^{-2}-8H_0v(v+|u|)^{-3}+\ldots
\end{equation*}
and
\begin{equation}
\label{eq:asympint}
\begin{split}
&\int_{-\infty}^{u_1}\partial_v\partial_V^2 r\partial_u\phi_0+\partial_u\partial_vr\partial_Vr\partial_V\phi_0\,du'\\
&=2H_0(v+|u_1|)^{-1}-4H_0v(v+|u_1|)^{-2}+\ldots\\
&=-2H_0(v+|u_1|)^{-1}+\ldots,
\end{split}
\end{equation}
where $\ldots$ indicates terms that are integrable in $v$. The initial asymptotics
\begin{equation*}
r\partial_V^2\phi_0(-\infty,v')=-2H_0\log v+\mathcal{O}(1)
\end{equation*}
ensure that the leading order term in (\ref{eq:asympint}) gets cancelled, so we can estimate
\begin{equation*}
\begin{split}
&\left|\int_{v_0}^{\infty}\left[\partial_v(r\partial_V^2\phi_0)(-\infty,v')-\int_{-\infty}^{u_1}\partial_v\partial_V^2 r\partial_u\phi_0+\partial_u\partial_v\partial_Vr\partial_V\phi_0\,du'\right]\,dv'\right|\\
&\leq C_0+ \left|\int_{v_0}^{\infty}\partial_v(M\partial_V^2\phi_0+2H_0\log v)\right|\\
& \leq C_0+\left|\lim_{v\to \infty}\left[M\partial_V^2\phi_0(-\infty,v)+2H_0\log v\right]-\partial_V^2\phi_0(-\infty,v_0)-2H_0\log v_0\right|\\
&\leq C_0,
\end{split}
\end{equation*}
where we used the assumption (\ref{eq:limitwelldefined}) to arrive at the final inequality above. Note that this assumption does not follow from the asymptotics in (\ref{eq:asympoticsphiinitial}), as they only imply that $M\partial_V^2\phi_0(-\infty,v)+2H_0\log v$ is uniformly bounded everywhere along $\mathcal{H}^+$.

It now follows that 
\begin{equation}
\label{eq:mainboundc2proof}
\left|\int_{v_0}^{\infty}\partial_v(r\partial_V^2\phi_0)(u_1,v')\,dv'\right|\leq C_0.
\end{equation}

By using (\ref{eq:mainboundc2proof}) and boundedness of $u_2-u_1$, it is straightforward to show that
\begin{equation*}
\left|\int_{v_0}^{\infty}\partial_v(r\partial_V^2\phi_0)(u_2,v')\,dv'\right|\leq C_0.
\end{equation*}
and conclude that $\partial_V^2\phi_0$ is continuously extendible beyond $\mathcal{CH}^+$.
\end{proof}

\section{Energy estimates}
\label{sec:eestimatesrn}
We now consider solutions to (\ref{eq:waveqkerr}) without any symmetry assumptions. \textbf{In this section $\phi$ will always denote a solution to} (\ref{eq:waveqkerr}) \textbf{corresponding to characteristic initial data from Proposition} \ref{prop:wellposedness}. In this case, weighted $L^2$ estimates derived by using the divergence identity from Section \ref{sec:vectorfieldmethod}, will replace the role of the weighted $L^1$ estimates from Proposition \ref{lm:l1}.

Consider the vector field $N_{p,q}$,
\begin{equation}
\label{def:mainvectorfieldN}
N_{p,q}=|u|^p\partial_u+v^q\partial_v,
\end{equation}
where $0\leq p,q\leq 2$.

We can express
\begin{equation*}
N_{p,q}=v(V)^q\left(1-\frac{M}{V}\right)^2\partial_V+|u(U)|^pU^2\partial_U,
\end{equation*}
where $U$ and $V$ are the regular double-null coordinates from Section \ref{sec:geom}.
We have that
\begin{align*}
v(V)^q\left(1-\frac{M}{V}\right)^2 \sim v^q(1+v)^{-2},\\
|u(U)|^pU^2  \sim |u|^p(1+|u|)^{-2},
\end{align*}
so $N_{p,q}$ is continuous at $\mathcal{H}^+$ and $\mathcal{CH}^+$ if and only if $p,q\leq 2$. Moreover, for $p=q=2$, $N_{p,q}$ is timelike everywhere in $D_{u_0,v_0}$.

By (\ref{eq:genexprKN}), we have that
\begin{equation}
\label{eq:exprbulk}
\begin{split}
K^N[\phi]&=r^{-1}(N^v-N^u)\partial_u\phi\partial_v\phi-\frac{1}{2}\left[\partial_vN^v+\partial_uN^u+\Omega^{-2}\partial_u(\Omega^2)N^u+\Omega^{-2}\partial_v(\Omega^2)N^v\right]\\
&=r^{-1}(N^v-N^u)\partial_u\phi \partial_v\phi-\frac{1}{2}\left[ \frac{2M}{r^3}(M-r)(N^u-N^v)+\partial_u N^u+\partial_vN^v\right]|\snabla \phi|^2.
\end{split}
\end{equation}
We let $N^v=v^q$ and $N^u=|u|^p$ and split
\begin{equation*}
K^{N_{p,q}}[\phi]=K^{N_{p,q}}_{\textnormal{null}}[\phi]+K^{N_{p,q}}_{\textnormal{angular}}[\phi],
\end{equation*}
where
\begin{align*}
K^{N_{p,q}}_{\textnormal{null}}[\phi]&:=r^{-1}(v^q-|u|^p)\partial_u\phi \partial_v\phi,\\
K^{N_{p,q}}_{\textnormal{angular}}[\phi]&:=-\left[ \frac{M}{r^3}(M-r)(|u|^p-v^q)-\frac{p}{2}|u|^{p-1}+\frac{q}{2}v^{q-1}\right]|\snabla \phi|^2.
\end{align*}
\begin{proposition}
\label{prop:mainenergyestimate}
Fix $p=2$. Take $0<q\leq 2$. There exists a constant $C=C(e,M,u_0,v_0,q)>0$ such that for all $H_u$ and $\uline{H}_v$ in ${D_{u_0,v_0}}$
\begin{equation*}
\begin{split}
&\int_{H_u} J^{N_{2,q}}[\phi]\cdot \partial_v+\int_{\uline{H}_v} J^{N_{2,q}}[\phi]\cdot \partial_u\\
&\leq C\left[ \int_{\mathcal{H}^+\cap\{v\geq v_0\}} J^{N_{2,q}}[\phi]\cdot\partial_v+ \int_{\uline{H}_{v_0}} J^{N_{2,q}}[\phi]\cdot \partial_u\right]=:CE_q[\phi].
\end{split}
\end{equation*}
\end{proposition}
\begin{proof}
By applying the divergence theorem in ${D_{u_0,v_0}}$, we can estimate
\begin{equation}
\label{eq:stokesaxikn}
\begin{split}
&\int_{H_u} J^{N_{2,q}}[\phi]\cdot \partial_v+\int_{\uline{H}_v} J^{N_{2,q}}[\phi]\cdot \partial_u\\
&= \int_{\mathcal{H}^+\cap\{v\geq v_0\}} J^{N_{2,q}}[\phi]\cdot \partial_v+ \int_{\uline{H}_{v_0}} J^{N_{2,q}}[\phi]\cdot \partial_u-\int_{{D_{u_0,v_0}}} K^{N_{p,q}}_{\textnormal{null}}[\phi]+K^{N_{p,q}}_{\textnormal{angular}}[\phi].
\end{split}
\end{equation}

We first consider $K^{N_{p,q}}_{\textnormal{null}}[\phi]$. By (\ref{eq:rnOmegainverseexp}) we have that,
\begin{equation*}
\Omega^2|K^{N_{p,q}}_{\textnormal{null}}[\phi]|\leq C(v+|u|)^{-2}(v^q-|u|^p)|{\partial_v}\phi||{\partial_u}\phi|.
\end{equation*}

By the Cauchy--Schwarz inequality, we can estimate for $\eta>0$,
\begin{equation*}
\begin{split}
v^q(v+|u|)^{-2}|{\partial_v}\phi||{\partial_u}\phi|&\leq v^q(v+|u|)^{-1-q-\eta} |u|^p (\partial_u\phi)^2+|u|^{-p}(v+|u|)^{q+\eta-3}v^q(\partial_v\phi)^2\\
&\leq v^q\sup_{-\infty\leq u'<u}\left[(v+|u'|)^{-1-q-\eta}\right]|u|^p(\partial_u\phi)^2\\
&\quad+|u|^{-p}\sup_{v_0<v'<v}\left[(v'+|u|)^{q+\eta-3}\right]v^q(\partial_v\phi)^2\\
&\leq Cv^{-1-\eta}|u|^p(\partial_u\phi)^2+C|u|^{q+\eta-p-3}v^q(\partial_v\phi)^2,
\end{split}
\end{equation*}
for $\eta<3-q$.

Similarly, for $u$ and $v$ reversed, we obtain
\begin{equation*}
|u|^p(v+|u|)^{-2}|{\partial_v}\phi||{\partial_u}\phi|\leq C|u|^{-1-\eta}v^q(\partial_v\phi)^2+Cv^{p+\eta-q-3}|u|^p(\partial_u\phi)^2,
\end{equation*}
for $\eta<3-p$.

We will now estimate $K^{N_{p,q}}_{\textnormal{angular}}$. We have that,
\begin{equation}
\label{eq:KNangularkerrn1}
\begin{split}
K^{N_{p,q}}_{\textnormal{angular}}&=-\frac{1}{2}\left[qv^{q-1}-p|u|^{p-1}+2Mr^{-3}(M-r)(v^q-|u|^p)\right]|\snabla\phi|^2\\
&\quad+(v^q+|u|^p)\mathcal{O}((v+|u|)^{-2})|\snabla\phi|^2.
\end{split}
\end{equation}
Moreover, by (\ref{eq:rnOmegainverseexp}) we have that
\begin{equation*} 
\frac{2M}{r^3}(M-r)=-\frac{2}{v+|u|}+\log (v+|u|)\mathcal{O}((v+|u|)^{-2}).
\end{equation*}
Consequently, we can rewrite (\ref{eq:KNangularkerrn1}) to obtain
\begin{equation}
\label{eq:KNangularkerrn2}
\begin{split}
K^{N_{p,q}}_{\textnormal{angular}}&=-\frac{1}{2}\left[\left(q-2\frac{v^{\frac{q}{2}}+|u|^{\frac{p}{2}}}{v+|u|}v^{1-\frac{q}{2}}\right)v^{q-1}+\left(2\frac{v^{\frac{q}{2}}+|u|^{\frac{p}{2}}}{v+|u|}|u|^{1-\frac{p}{2}}-p\right)|u|^{p-1}\right]|\snabla\phi|^2\\
&\quad+(v^q+|u|^p)\log(v+|u|)\mathcal{O}((v+|u|)^{-2})|\snabla\phi|^2.
\end{split}
\end{equation}
Note first of all that, if $0\leq p<2$, the terms inside square brackets become positive in the region $|u|>v$, as we approach $\mathcal{H}^+$, so $K^{N_{p,q}}_{\textnormal{angular}}$ becomes negative and we are not be able to control it. We therefore restrict to $p=2$. 

If $p=2$ and $q<2$, the term inside the square brackets is negative and we can estimate
\begin{equation*}
K^{N_{p,q}}_{\textnormal{angular}}\geq Cv^{q-1}|\snabla\phi|^2+(v^q+|u|^2)\log(v+|u|)\mathcal{O}((v+|u|)^{-2})|\snabla\phi|^2.
\end{equation*}

If $p=2$ and $q=2$, both the first term and second term multiplying $|\snabla\phi|^2$ in (\ref{eq:KNangularkerrn2}) vanish, so we can estimate
\begin{equation*}
K^{N_{p,q}}_{\textnormal{angular}}=(v^2+|u|^2)\log(v+|u|)\mathcal{O}((v+|u|)^{-2})|\snabla\phi|^2.
\end{equation*}

If we fix $p=2$, we can therefore estimate for all $0\leq q\leq 2$,
\begin{equation*}
\begin{split}
\Omega^2 K^{N_{p,q}}_{\textnormal{angular}}&\geq (v^q+|u|^2)\log(v+|u|)\mathcal{O}((v+|u|)^{-2})\Omega^2|\snabla\phi|^2\\
&\geq -C_{\epsilon}|u|^{-2+\epsilon}v^q\Omega^2|\snabla\phi|^2-Cv^{-2+\epsilon}|u|^2\Omega^2|\snabla\phi|^2,
\end{split}
\end{equation*}
with arbitrarily small $\epsilon>0$ and $C_{\epsilon}=C_{\epsilon}(M,u_0,v_0,\epsilon)>0$. We will fix $0<\epsilon<1$.

Putting all the above estimates together, we find that for $0\leq q\leq 2$,
\begin{equation*}
\begin{split}
-\Omega^2 (K^{N_{2,q}}_{\textnormal{angular}}+K^{N_{2,q}}_{\textnormal{null}})&\leq C\Big[(v^{-1-\eta}+v^{\eta-q-1})|u|^2(\partial_u\phi)^2+(|u|^{-1-\eta}+|u|^{q+\eta-5})v^q(\partial_v\phi)^2\\
&\quad+|u|^{-2+\epsilon}v^q\Omega^2|\snabla\phi|^2+v^{-2+\epsilon}|u|^2\Omega^2|\snabla\phi|^2\Big].
\end{split}
\end{equation*}
We can now apply Lemma \ref{lm:gronwall} with
\begin{align*}
A&= \int_{\mathcal{H}^+\cap\{v\geq v_0\}} J^{N_{2,q}}[\phi]\cdot \partial_v+ \int_{\uline{H}_{v_0}} J^{N_{2,q}}[\phi]\cdot \partial_u,\\
f(u,v)&=\int_{H_{u}}v^q(\partial_v\phi)^2+|u|^2\Omega^2|\snabla\phi|^2,\\
g(u,v)&=\int_{\underline{H}_v}|u|^2(\partial_u\phi)^2+v^q\Omega^2|\snabla\phi|^2,\\
h(u)&=|u|^{-1-\eta}+|u|^{q+\eta-5}+|u|^{-2+\epsilon},\\
k(v)&=v^{-1-\eta}+v^{\eta-q-1}+v^{-2+\epsilon}.
\end{align*}
If we take $q>0$, $h$ and $k$ are integrable for $0<\eta<\min\{q,1\}$, and we arrive at the estimate in the proposition.
\end{proof}

Since $\square_g(O^k\phi)=0$ for all $k\in \N$ if $\square_g\phi=0$, the above energy estimate also holds for the angular derivatives $O^k\phi$. We can also consider the higher-order null derivative operator $X_{m,n}$, which is defined as follows:
\begin{equation*}
X_{m,n}f:=\partial_u^m\partial_v^nf.
\end{equation*}

As $\partial_u$ and $\partial_v$ are not Killing vector fields, we cannot conclude that $\square_g(X_{m,n}\phi)=0$, but we can still obtain energy estimates for $X_{m,n}\phi$ by controlling the additional error terms.

\begin{proposition}
\label{prop:highordeestimate}
Let $m,n\geq 0$, fix $p=2$ and take $0<q\leq 2$. Then there exists a constant $C=C(M,u_0,v_0,q,m,n)>0$ such that for all $l\geq 0$ and $H_u$ and $\uline{H}_v$ in ${D_{u_0,v_0}}$,
\begin{equation*}
\begin{split}
&\int_{H_u} J^{N_{2,q}}[X_{m,n}O^l\phi]\cdot \partial_v+\int_{\uline{H}_v} J^{N_{2,q}}[X_{m,n}O^l\phi]\cdot \partial_u\\
&\leq C\sum_{m'=\min\{1,m\}}^m\sum_{n'=\min\{1,n\}}^n\sum_{s=0}^{\max\{m-m',n-n'\}}\Bigg[ \int_{\mathcal{H}^+\cap\{v\geq v_0\}} J^{N_{2,q}}[X_{m',n'}O^sO^l\phi]\cdot \partial_v\\
&\quad+ \int_{\uline{H}_{v_0}} J^{N_{2,q}}[X_{m',n'}O^sO^l\phi]\cdot \partial_u\Bigg]=:CE_{q;(m,n,l)}[\phi].
\end{split}
\end{equation*}
\end{proposition}
\begin{proof}
Without loss of generality we take $l=0$, as we can replace $\phi$ by $O^l\phi$ with $l>0$, since $O_i$, $i=1,2,3$ are Killing vector fields.

We apply the product rule for higher-order differentiation to obtain
\begin{equation*}
\begin{split}
2\Omega^2 \square_g(X_{m,n}\phi)&=X_{m,n}(\square_g\phi)+2 \sum_{k=1}^m\sum_{l=1}^n \binom{m}{k}\binom{n}{l} \partial_u^k\partial_v^l(r^{-1}\partial_v r)(\partial_v-\partial_u)(X_{m-k,n-l}\phi)\\
&\quad+2 \sum_{k=1}^m\sum_{l=1}^n \binom{m}{k}\binom{n}{l}\partial_u^k\partial_v^l(r^{-2}\Omega^2)X_{m-k,n-l}(\mathring{\slashed{\Delta}}\phi),
\end{split}
\end{equation*}
where $\mathring{\slashed{\Delta}}$ denotes the Laplacian on $\s^2$, the round sphere of radius 1.

In particular, we can estimate 
\begin{equation*}
\begin{split}
2\Omega^2 |\square_g(X_{m,n}\phi)|&\leq C\sum_{k=1}^m\sum_{l=1}^n (v+|u|)^{-(k+l+2)}\left[|\partial_v X_{m-k,n-l}\phi|+|\partial_u X_{m-k,n-l}\phi|+|\mathring{\slashed{\Delta}}X_{m-k,n-l}\phi|\right].
\end{split}
\end{equation*}
Consequently, we can estimate by using Cauchy--Schwarz,
\begin{equation*}
\begin{split}
2\Omega^2\left|\mathcal{E}^{N_{p,q}}[X_{m,n}\phi]\right|&=2\Omega^2\left||u|^p\partial_u(X_{m,n}\phi)+v^q\partial_v(X_{m,n}\phi)\right||\square_g\phi|\\
&\leq C (v+|u|)^{-1-\eta}|u|^p(\partial_u(X_{m,n}\phi))^2+C (v+|u|)^{-1-\eta}v^q(\partial_v(X_{m,n}\phi))^2\\
&\quad+ C (|u|^p+v^q)\sum_{k=1}^m\sum_{l=1}^n (v+|u|)^{-(2k+2l+3-\eta)}\Big[|\partial_v X_{m-k,n-l}\phi|^2+|\partial_u X_{m-k,n-l}\phi|^2\\
&\quad+|\mathring{\slashed{\Delta}}X_{m-k,n-l}\phi|^2\Big].
\end{split}
\end{equation*}

The following equality holds for the angular momentum operators $O_i$:
\begin{equation*}
\int_{\s^2}(\mathring{\slashed{\Delta}}f)^2\,d\mu_{\slashed{g}_{\s^2}}= r^2\sum_{s=0}^1|\snabla{O^sf}|^2\,d\mu_{\slashed{g}_{\s^2}}.
\end{equation*}
We can therefore estimate the spacetime integral of $\left|\mathcal{E}^{N_{p,q}}[X_{m,n}\phi]\right|$ by the flux terms of $J^{N_{p,q}}[X_{m-k,n-l}\phi]$ with $k,l\geq 0$ and $J^{N_{p,q}}[X_{m-k,n-l}O^s\phi]$, with $k\geq 1$ or $l\geq 1$, multiplied by an integral function in $u$ or $v$, as in Proposition \ref{prop:mainenergyestimate}. We subsequently apply Lemma \ref{lm:gronwall}.

We obtain the energy estimate in the proposition by induction on $m$ and $n$, in order to estimate the energy fluxes with additional $O^s$ derivatives. In the induction step, we estimate the fluxes of $J^{N_{p,q}}[X_{m+1,n}\phi]$, $J^{N_{p,q}}[X_{m,n+1}\phi]$ and $J^{N_{p,q}}[X_{m+1,n+1}\phi]$ by the fluxes of $J^{N_{p,q}}[X_{m,n}O^s\phi]$, where $s\leq 1$.
\end{proof}
We have now proved Theorem \ref{thm:eestimatesint}.

\begin{corollary}
\label{cor:highordeestimate}
Let $X_{m,n}=\partial_u^m\partial_v^n$ for $m,n\geq 0$ and let $\Sigma'\subset D_{u_0,v_0}$ be a spacelike hypersurface intersecting $\mathcal{CH}^+$, with unit normal $n_{\Sigma'}$. Then there exists a constant $C=C(M,\Sigma',u_0,v_0,m,n)>0$ such that for all $l\geq 0$,
\begin{equation*}
\int_{\Sigma'} J^{n_{\Sigma'}}[X_{m,n}O^l\phi]\cdot n_{\Sigma'}\leq CE_{2;(m,n,l)}[\phi].
\end{equation*}
\end{corollary}
\begin{proof}
We have that
\begin{equation*}
\int_{\Sigma'} J^{n_{\Sigma'}}[X_{m,n}O^l\phi]\cdot n_{\Sigma'} \leq C\int_{\Sigma}(\partial X_{m,n}O^l\phi)^2\leq C \int_{\Sigma'} J^{N_{2,2}}[X_{m,n}O^l\phi]\cdot n_{\Sigma'},
\end{equation*}
because $N_{2,2}$ is uniformly timelike everywhere along $\Sigma'$, including at $\mathcal{CH}^+$.

We then apply the divergence theorem to $J^{N_{2,2}}$ in $J^-(\Sigma')\cap D_{u_0,v_0}$ and estimate $K^{N_{2,2}}$ as in Proposition \ref{prop:highordeestimate}.
\end{proof}

\section{Pointwise estimates}
\label{sec:pointwise}
\subsection{Uniform boundedness of $\partial_u^m\partial_v^nO^l\phi$}
\label{sec:unifboundcomm}

Uniform boundedness of $\phi$ and its $u$, $v$ and angular derivatives follows easily from the energy estimates derived in the previous section. As in the previous section, $\phi$ always denotes a solution to (\ref{eq:waveqkerr}).

\begin{proposition}
\label{prop:boundpsirn}
Let $0<q<2$ and $m,n\geq 0$. There exists a constant $C(M,u_0,v_0,q)>0$ such that
\begin{equation}
\label{eq:boundpsi}
\begin{split}
||X_{m,n}O^l\phi||_{L^{\infty}(H_u)}^2 &\leq C\sum_{|k|\leq 2} \sup_{v_0\leq {v'} \leq \infty}\int_{\s^2} (O^k X_{m,n}O^l\phi)^2\,d\mu_{\s^2}(-\infty,{v'})+ C|u|^{-1}\sum_{|k|\leq 2}E_{q;(m,n,l)}[O^k\phi].
\end{split}
\end{equation}
\end{proposition}
\begin{proof}
Let $(u,v,\theta,\varphi)\in D_{u_0,v_0}$. We can then apply the fundamental theorem of calculus, while keeping $(\theta,\varphi)$ fixed, to obtain
\begin{equation*}
|\phi|(u,v,\theta,\varphi)\leq |\phi|(-\infty,v,\theta,\varphi)+\int_{-\infty}^u |\partial_{u}\phi|({u'},v,\theta,\varphi)\,du'.
\end{equation*}
In order to convert the $L^1$ norm above into an $L^2$ norm, we apply Cauchy--Schwarz
\begin{equation}
\label{mainestboundpsirn}
\phi^2(u,v,\theta,\varphi)\leq \phi^2(-\infty,v,\theta,\varphi)+ \int_{-\infty}^u |u'|^{-2}\,d{u'} \int_{-\infty}^u |{u'}|^{2} (\partial_u \phi)^2({u'},v,\theta,\varphi)\,d{u'},
\end{equation}
where $p>1$.

Now we integrate both sides over $\s^2$ and commute with angular momentum operators $O_i$, so that we can apply a standard Sobolev inequality on $\s^2$ to conclude
\begin{equation}
\label{preboundpsi}
\phi^2(u,v,\theta,\varphi)\leq C\sum_{|k|\leq 2} \int_{\s^2} (O^k \phi)^2\,d\mu_{\s^2}(-\infty,v)+ C |u|^{-1}\int_{\uline{H}_v} |u|^{2} (\partial_u O^k \phi)^2,
\end{equation}
where $C>0$ is a constant. The weighted $L^2(\uline{H}_v)$ integral on the right-hand side can be bounded by a suitably energy flux,
\begin{equation*}
\int_{\uline{H}_v}  |u|^{2} (\partial_u\phi)^2\leq \int_{\uline{H}_v}J^{N_{2,q}}[\phi]\cdot \partial_u,
\end{equation*}
with $q\geq 0$. Combining (\ref{preboundpsi}) with Proposition \ref{prop:mainenergyestimate} immediately gives uniform pointwise boundedness of $\phi$ in ${D_{u_0,v_0}}$ Similarly, we obtain uniform pointwise boundedness of $X_{m,n}\phi$ by applying the energy estimate in Proposition \ref{prop:highordeestimate} instead.
\end{proof}

We have now proved Theorem \ref{thm:pointwiseestho}.

\subsection{$C^0$-extendibility of $\partial_u^m\partial_v^nO^l\phi$}
We can use Proposition \ref{prop:highordeestimate} to show that $\partial_u^m\partial_v^n O^l\phi$ can be extended as continuous functions beyond the Cauchy horizon $\mathcal{CH}^+$. We introduce the second-order differential operator $X_{m,n,l}$, defined by
\begin{equation*}
X_{m,n,l}f:=\partial_u^m\partial_v^nO^lf.
\end{equation*}

\begin{proposition}
\label{prop:C0extension}
Let the initial data for $\phi$ satisfy
\begin{equation*}
\sum_{|k|\leq 2}E_{q;(m,n,l)}[O^k\phi]<\infty,
\end{equation*}
for $q>1$.

Let $x_{\mathcal{CH}^+}$ be a point on $\mathcal{CH}^+$. Then, for $x\in {D_{u_0,v_0}}$,
\begin{equation*}
\lim_{x\to x_{\mathcal{CH}^+}}(X_{m,n,l}\phi)(x)
\end{equation*}
is well-defined, so $X_{m,n,l}\phi$ can be extended as a $C^0$ function to the region beyond $\mathcal{CH}^+$.
\end{proposition}
\begin{proof}
Consider a sequence of points $x_k$ in ${D_{u_0,v_0}}\setminus \mathcal{H}^+$, such that $\lim_{k\to \infty}x_k=x_{\mathcal{CH}^+}$. The sequence $\{x_k\}$ is in particular a Cauchy sequence. We will show that the sequence of points $(X_{m,n,l}\phi)(x_k)$ must also be a Cauchy sequence, from which it follows immediately that the sequence converges to a finite number as $k\to \infty$.

For simplicity, we will take $m-n=l=0$, but the steps of the proof can be repeated for the general case. Denote $x_k=(u_k,V_k,\theta_k,\varphi_k)$. Let $l>k$, then
\begin{equation*}
\begin{split}
|\phi(x_l)-\phi(x_k)|^2&\leq |\phi(u_l,V_k,\theta_k,\varphi_k)-\phi(u_k,V_k,\theta_k,\varphi_k)|^2+|\phi(u_k,V_l,\theta_k,\varphi_k)-\phi(u_k,V_k,\theta_k,\varphi_k)|^2\\
&\quad+|\phi(u_k,V_k,\theta_l,\varphi_k)-\phi(u_k,V_k,\theta_k,\varphi_k)|^2+|\phi(u_k,V_k,\theta_k,\varphi_l)-\phi(u_k,V_k,\theta_k,\varphi_k)|^2
\end{split}
\end{equation*}
By the fundamental theorem of calculus, a Sobolev inequality on $\s^2$ and Cauchy--Schwarz, we can estimate for $q>0$
\begin{equation*}
\begin{split}
\left|\phi(u_l,V_k,\theta_k,\varphi_k)-\phi(u_k,V_k,\theta_k,\varphi_k)\right|^2&\leq C\sum_{|s|\leq 2}\left|\int_{u_k}^{u_l}\int_{\s^2}u^{1+\epsilon}(\partial_uO^s\phi)^2(u,V_k,\theta_k,\varphi_k)\,d\mu_{\s^2}du\right|\\
&\leq C \sum_{|s|\leq 2}\int_{\uline{H}_{v(V_k)}}J^N_{2,q}[O^s\phi]\cdot \partial_u.
\end{split}
\end{equation*}

Similarly, we find that for $q>1$,
\begin{equation*}
\begin{split}
&|\phi(u_k,V_l,\theta_k,\varphi_k)-\phi(u_k,V_k,\theta_k,\varphi_k)|^2\\
&\leq C\left|(V_l-M)^{q-1}-(V_k-M)^{q-1}\right|\sum_{|s|\leq 2}\left|\int_{V_k}^{V_l}\int_{\s^2}(V-M)^{2-q}(\partial_VO^s\phi)^2(u,V,\theta_k,\varphi_k)\,d\mu_{\s^2}dV\right|\\
&\leq C\left|(V_l-M)^{q-1}-(V_k-M)^{q-1}\right|\sum_{|s|\leq 2}\left|\int_{v(V_k)}^{v(V_l)}\int_{\s^2}v^q(\partial_vO^s\phi)^2(u,v,\theta_k,\varphi_k)\,d\mu_{\s^2}dv\right|\\
&\leq C \left|(V_l-M)^{q-1}-(V_k-M)^{q-1}\right|\sum_{|s|\leq 2}\int_{H_{u_k}}J^N_{2,q}[O^s\phi]\cdot \partial_v,
\end{split}
\end{equation*}
where we used that $(V-M)^{2-q}\sim v^{q-2}$ and $(\partial_VO^s\phi)^2dV\sim  v^2(\partial_vO^s\phi)^2dv$.

Finally, we can estimate by Cauchy--Schwarz on $\s^2$,
\begin{equation*}
\begin{split}
&|\phi(u_k,v_l,\theta_l,\varphi_k)-\phi(u_k,v_k,\theta_k,\varphi_k)|^2+|\phi(u_k,v_l,\theta_k,\varphi_l)-\phi(u_k,v_k,\theta_k,\varphi_k)|^2\\
&\leq C\int_{\s^2}|\snabla\phi|^2(u_k,v_k,\theta,\varphi)\,d\mu_{\s^2}\leq C \int_{\uline{H}_{v(V_k)}}J^N_{2,q}[O \phi]\cdot \partial_u,
\end{split}
\end{equation*}
where we need $q>0$.

By the above estimates it follows that $\phi(p_k)$ must also be a Cauchy sequence if the energies on the right-hand sides are finite.

We use Proposition \ref{prop:highordeestimate} with $q>1$ to estimate the energies on the right-hand side by initial energies. We also use that for $q>1$ we can estimate, by applying the fundamental theorem of calculus together with Cauchy--Schwarz,
\begin{equation*}
\sum_{|s|\leq 2} \sup_{v_0\leq {v'} \leq \infty}\int_{\s^2} (O^sX_{m,n,l}\phi )^2\leq C\sum_{|s|\leq 2} E_{q,(m,n,l)}[O^s\phi]. \qedhere
\end{equation*}
\end{proof}
We have now proved Theorem \ref{thm:C0extension}.

\subsection{Decay of $\partial_v\phi$}
\label{sec:decaydernosymm}
Consider the function $\mathcal{\phi_{\mathcal{H}^+}}:\mathcal{M}\cap {D_{u_0,v_0}}\to \R$ defined by
\begin{equation*}
\phi_{\mathcal{H}^+}(u,v,\theta,\varphi):=\phi(-\infty,v,\theta,\varphi).
\end{equation*}
In particular, ${\partial_u}\phi_{\mathcal{H}^+}=0$.

We consider $\psi:=\phi-\phi_{\mathcal{H}^+}$. We can obtain pointwise decay in $\psi$ with respect to $|u|$ and use the wave equation (\ref{eq:waveqkerr}) to obtain decay of $|{\partial_v}\phi|$ in $v$.
\begin{proposition}
\label{improvedpointwisedecaykn}
Let $0<q\leq 2$. There exists a constant $C=C(M,u_0,v_0,q)>0$, such that for all $H_u$ and $\uline{H}_v$ in ${D_{u_0,v_0}}$,
\begin{equation*}
\begin{split}
&\int_{H_u} J^{N_{2,q}}[\psi]\cdot {\partial_v}+\int_{\uline{H}_v\cap\{u'\leq u\}} J^{N_{2,q}}[\psi]\cdot {\partial_u}\\
&\leq C|u|^{-s} \int_{\mathcal{H}^+\cap\{v\geq v_0\}} v^{s+\epsilon}\left[({\partial_v}\phi)^2+|\snabla\phi|^2+|\snabla^2\phi|^2\right]+ C\int_{\uline{H}_{v_0}\cap\{u'\leq u\}} J^{N_{2,q}}[\psi]\cdot {\partial_u},
\end{split}
\end{equation*}
with $0\leq s\leq 1$ and $\epsilon>0$ arbitrarily small.

In particular, using that the initial data for $\phi$ must satisfy for any $0\leq s\leq 1$
\begin{equation*}
\int_{-\infty}^u\int_{\s^2} {u'}^2(\partial_u\phi)^2+{u'}^{-2}|\snabla \psi|^2\,d\mu_{\s^2}du'\Big|_{v=v_0}\leq D_s|u|^{-s},
\end{equation*}
where $D_s$ is a constant, we have that
\begin{equation*}
\begin{split}
&\int_{H_u} J^{N_{2,q}}[\psi]\cdot {\partial_v}+\int_{\uline{H}_v\cap\{u'\leq u\}} J^{N_{2,q}}[\psi]\cdot {\partial_u}\\
&\leq C|u|^{-s} \left[\int_{\mathcal{H}^+\cap\{v\geq v_0\}} v^{s+\epsilon}\left[({\partial_v}\phi)^2+|\snabla\phi|^2+|\snabla^2\phi|^2\right]+ D_s\right]\\
&:=C|u|^{-s}\tilde{E}_{s;\epsilon}[\phi].
\end{split}
\end{equation*}
\end{proposition}
\begin{proof}
We have that
\begin{equation*}
2\Omega^2\square_g\psi=-2\Omega^2\square_g\phi|_{\mathcal{H}^+}=-2r^{-1}\partial_u r{\partial_v}\phi_{\mathcal{H}^+}+2\Omega^2\slashed{\Delta}\phi_{\mathcal{H}^+}.
\end{equation*}
Consequently, we can estimate,
\begin{equation*}
2\Omega^2|\square_g\psi|\leq C(v+|u|)^{-2}\left(|{\partial_v}\phi_{\mathcal{H}^+}|+|\snabla\phi_{\mathcal{H}^+}|+|\snabla^2\phi_{\mathcal{H}^+}|\right).
\end{equation*}
By applying the divergence theorem in $D_{u_0,v_0}$ we obtain the following error term:
\begin{equation*}
\begin{split}
&\left|\int_{\mathcal{M}\cap {D_{u_0,v_0}}\cap\{u'\leq u\}}\mathcal{E}^{N_{p,q}}[\psi]\right|\\
&\leq C\int_{-\infty}^u \int_{v_0}^v \int_{\s^2}(v+|u|)^{-2}(|u|^p|{\partial_u}\psi|+v^q|{\partial_v}\psi|)\left(|{\partial_v}\phi_{\mathcal{H}^+}|+|\snabla\phi_{\mathcal{H}^+}|+|\snabla^2\phi_{\mathcal{H}^+}|\right)\,2\Omega^2r^2d\mu_{\s^2}dvdu.
\end{split}
\end{equation*}
By Cauchy--Schwarz, we can estimate for $\eta>0$
\begin{equation*}
\begin{split}
&(v+|u|)^{-2}(|u|^p|{\partial_u}\psi|+v^q|{\partial_v}\psi|)\left(|{\partial_v}\phi_{\mathcal{H}^+}|+|\snabla\phi_{\mathcal{H}^+}|+|\snabla^2\phi_{\mathcal{H}^+}|\right)\lesssim v^{-1-\eta}|u|^p(\partial_u\psi)^2\\
&\quad+|u|^{-1-\eta'}v^q(\partial_v\psi)^2+(v+|u|)^{-4}(|u|^p{v}^{1+\eta}+v^q|u|^{1+\eta})\left[({\partial_v}\phi_{\mathcal{H}^+})^2+|\snabla\phi_{\mathcal{H}^+}|^2+|\snabla^2\phi_{\mathcal{H}^+}|^2\right]
\end{split}
\end{equation*}
We further estimate for $0\leq s\leq 1$,
\begin{equation*}
\begin{split}
&(v+|u|)^{-4}(|u|^p{v}^{1+\eta}+v^q|u|^{1+\eta})\left[({\partial_v}\phi_{\mathcal{H}^+})^2+|\snabla\phi_{\mathcal{H}^+}|^2+|\snabla^2\phi_{\mathcal{H}^+}|^2\right]\\
&\leq C |u|^{-1-s} (v^{p-2+s+\eta}+v^{q-2+s+\eta})\left[({\partial_v}\phi_{\mathcal{H}^+})^2+|\snabla\phi_{\mathcal{H}^+}|^2+|\snabla^2\phi_{\mathcal{H}^+}|^2\right]
\end{split}
\end{equation*}
Hence,
\begin{equation*}
\begin{split}
&\int_{\mathcal{M}\cap {D_{u_0,v_0}}\cap\{u'\leq u\}}|u|^{-1-s} (v^{p-2+s+\eta}+v^{q-2+s+\eta})\left[({\partial_v}\phi_{\mathcal{H}^+})^2+|\snabla\phi_{\mathcal{H}^+}|^2+|\snabla^2\phi_{\mathcal{H}^+}|^2\right]\\
&\leq C|u|^{-s}\int_{\mathcal{H}^+\cap \{v'\geq v_0\}}(v^{p-2+s+\eta}+v^{q-2+s+\eta})\left(({\partial_v}\phi)^2+|\snabla\phi|^2+|\snabla^2\phi|^2\right),
\end{split}
\end{equation*}
where we used that
\begin{equation*}
({\partial_v}\phi_{\mathcal{H}^+})^2+|\snabla\phi_{\mathcal{H}^+}|^2+|\snabla^2\phi_{\mathcal{H}^+}|^2 \sim ({\partial_v}\phi)|_{\mathcal{H}^+}^2+|\snabla\phi|^2|_{\mathcal{H}^+}+|\snabla^2\phi|^2|_{\mathcal{H}^+}.
\end{equation*}
The remaining terms in $\mathcal{E}^{N_{p,q}}[\psi]$ and the terms in $K^{N_{p,q}}[\psi]$ can be estimated as in Proposition \ref{prop:mainenergyestimate}, applying Lemma \ref{lm:gronwall}.
\end{proof}

\begin{proposition}
\label{cor:improvedpointbound}
Let $0\leq s\leq 1$. There exists a constant $C=C(M,v_0,u_0,s)>0$ such that,
\begin{equation*}
\begin{split}
\int_{S^2_{u,v}}|\snabla\psi|^2+|\snabla^2\psi|^2\,d\mu_{\slashed{g}}&\leq C|u|^{-(s+1)}\sum_{k=1}^2\tilde{E}_{s;\epsilon}[O^k\phi].
\end{split}
\end{equation*}
\end{proposition}
\begin{proof}
Apply the proof of Proposition \ref{prop:boundpsirn}, together with the energy estimates for $\psi$ in Proposition \ref{improvedpointwisedecaykn}.
\end{proof}

We can rewrite the wave equation (\ref{eq:waveqkerr}) as a transport equation:
\begin{equation}
\label{eq:transpeqRN}
\partial_u(r\partial_{v}\phi)=-\partial_{v}r\partial_{u}\phi+r\partial_{v}r \slashed{\Delta}\phi.
\end{equation}
Consequently, after integrating in $u$ and over $\s^2$, we can estimate
\begin{equation}
\label{eq:dVpsiRN}
\begin{split}
\int_{\s^2}(\partial_v\phi)^2(u,v)\,d\mu_{\s^2}&\leq \int_{\s^2}(\partial_v\phi)^2(-\infty,v)\,d\mu_{\s^2}+\int_{\s^2}\left(\int_{-\infty}^u \partial_{v}r\partial_{u}\phi\,d{u'}\right)^2\,d\mu_{\s^2}\\
&\quad+\int_{\s^2}\left(\int_{-\infty}^u r\partial_{v}r \slashed{\Delta}\psi\,d{u'}\right)^2\,d\mu_{\s^2}\\
&\quad+\int_{\s^2}\left(\int_{-\infty}^u r\partial_{v}r \slashed{\Delta}\phi_{\mathcal{H}_+}\,d{u'}\right)^2\,d\mu_{\s^2}.
\end{split}
\end{equation}

\begin{proposition}
\label{prop:boundderVpsi}
For any $\epsilon>0$, there exists a constant $C=C(M,v_0,u_0,\epsilon)>0$ such that,
\begin{equation*}
\begin{split}
\int_{S^2_{u,v}}({\partial_v}\phi)^2(u,v,\theta,{\varphi})\,d\mu_{\slashed{g}}&\leq \int_{S^2_{-\infty,v}}({\partial_v}\phi)^2\,d\mu_{\slashed{g}}+C(v+|u|)^{-4}|u|^{-\eta}E_{\eta}[\phi]\\
&\quad+C(v+|u|)^{-4+\epsilon}\sum_{k=1}^2\tilde{E}_{1-\epsilon;0}[O^k\phi]\\
&\quad+C(v+|u|)^{-2}\sum_{k=1}^2 \int_{S^2_{-\infty,v}}(O^k\phi)^2\,d\mu_{\slashed{g}}.
\end{split}
\end{equation*} 

Moreover, we can estimate
\begin{equation*}
\begin{split}
\int_{S^2_{u,v}}({\partial_v}\phi)^2(u,v,\theta,{\varphi})\,d\mu_{\slashed{g}}&\leq \int_{S^2_{-\infty,v}}({\partial_v}\phi)^2\,d\mu_{\slashed{g}}+C(v+|u|)^{-4}|u|^{-\eta}E_{\eta}[\phi]\\
&\quad+C(v+|u|)^{-4}\log\left(\frac{v+|u|}{|u|}\right)\sum_{k=1}^2\tilde{E}_{1;\epsilon}[O^k\phi]\\
&\quad+C(v+|u|)^{-2}\sum_{k=1}^2 \int_{S^2_{-\infty,v}}(O^k\phi)^2\,d\mu_{\slashed{g}},
\end{split}
\end{equation*}
for $\epsilon>0$ arbitrarily small.

In particular, $\phi$ can be extended as a $C^{0,\alpha}$ function beyond $\mathcal{CH}^+$, for $\alpha<1$.
\end{proposition}
\begin{proof}
We apply the energy estimate in Proposition \ref{prop:mainenergyestimate} and the improved $L^2(\s^2)$ estimate in Proposition \ref{cor:improvedpointbound} to estimate the right-hand side of (\ref{eq:dVpsiRN}). We need in particular the following integral estimates:
\begin{align*}
\int_{-\infty}^u|u'|^{-1}(v+|u'|)^{-2}\,du'&\leq C(v+|u|)^{-2}\log\left(\frac{v+|u|}{|u|}\right),\\
\int_{-\infty}^u|u'|^{-s}(v+|u'|)^{-2}\,du'&\leq C(v+|u|)^{-2+(1-s+\eta)}|u|^{-\eta},
\end{align*}
for $0<s<1$ and $\eta>0$ arbitrarily small, to arrive at the estimates in the proposition.

We can replace $\phi$ by $O\phi$ and $O^2\phi$ in the above arguments to obtain, after applying a standard Sobolev inequality on the spheres $S^2_{u,v}$, an $L^{\infty}$ estimate for $\partial_v\phi$. We have in particular that $v^{\alpha+1}\partial_v\phi$ is uniformly bounded in $\mathcal{M}\cap {D_{u_0,v_0}}$, for suitable initial data, for all $0\leq \alpha<1$. It immediately follows that $\phi$ can be extended as a $C^{0,\alpha}$ function beyond $\mathcal{CH}^+$.
\end{proof}
We have now proved Theorem \ref{thm:estpartialVphi} (i).

\bibliography{mybib2}
\bibliographystyle{amsplaininitials}
\end{document}